\documentclass[article]{IEEEtran}
\usepackage{amsmath,amssymb,latexsym,cite}
\usepackage{epsf, pstricks,pgf, xcolor}

\def\hmax{h_{\text{MAX}}}

\def\Hb{\mathbf{H}}

\def\hh{\hat{h}}
\def\Hhb{\mathbf{\hat{H}}}

\def\Hs{\mathsf{H}}
\def\Hhs{\mathsf{\hat{H}}}

\def\Cbb{\mathbb{C}}

\def\P{\mathbb{P}}
\def\E{\mathbb{E}}

\newtheorem{theorem}{Theorem}

\newtheorem{lemma}{Lemma}
\newtheorem{corollary}{Corollary}

\newtheorem{definition}{Definition}

\newtheorem{remark}{Remark}

\begin{document}

\title{Ergodic Interference Alignment}
\author{Bobak Nazer, \IEEEmembership{Member, IEEE}, Michael Gastpar, \IEEEmembership{Member, IEEE}, \\Syed Ali Jafar, \IEEEmembership{Senior Member, IEEE}, Sriram Vishwanath, \IEEEmembership{Senior Member, IEEE}
\thanks{B. Nazer is with the Department of Electrical and Computer Engineering, Boston University, Boston, MA, 02215 USA (email: bobak@bu.edu). M. Gastpar is with the Department of Electrical Engineering and Computer
Sciences, University of California, Berkeley, CA 94720 USA, and with the School of Computer and Communication Sciences, Ecole Polytechnique F\'ed\'erale (EPFL), 1015 Lausanne, Switzerland (e-mail: gastpar@eecs.berkeley.edu). S. A. Jafar is with the Department of Electrical Engineering and Computer Science, University of California, Irvine, Irvine, CA, 92697-2625 (email: syed@uci.edu). S. Vishwanath is with the Department of Electrical and Computer Engineering, University of Texas, Austin, Austin, TX, 78712, USA (email: sriram@ece.utexas.edu).}
\thanks{B. Nazer and M. Gastpar were supported by NSF grants CCR-0347298, CNS-0627024, and CCF-0830428. M. Gastpar was also supported by the European ERC Starting Grant 259530-ComCom. S. A. Jafar was supported by NSF grant CCF-0830809, ONR YIP grant N00014-08-1-0872, and ONR grant N00014-12-1-0067. S. Vishwanath was supported by  ARO YIP grant 52491CI. The material in this paper was presented in part at the IEEE International Symposium on Information Theory, Seoul, South Korea, July 2009 and at the 47th Annual Allerton Conference on Communications, Control, and Computing, September 2009.}}
\markboth{To appear IEEE Trans. Info. Theory.}{~}

\maketitle

\begin{abstract}
This paper develops a new communication strategy, \textit{ergodic interference alignment}, for the $K$-user interference channel with time-varying fading. At any particular time, each receiver will see a superposition of the transmitted signals plus noise. The standard approach to such a scenario results in each transmitter-receiver pair achieving a rate proportional to $1/K$ its interference-free ergodic capacity. However, given two well-chosen time indices, the channel coefficients from interfering users can be made to exactly cancel. By adding up these two observations, each receiver can obtain its desired signal without any interference. If the channel gains have independent, uniform phases, this technique allows each user to achieve at least $1/2$ its interference-free ergodic capacity at any signal-to-noise ratio. Prior interference alignment techniques were only able to attain this performance as the signal-to-noise ratio tended to infinity. Extensions are given for the case where each receiver wants a message from more than one transmitter as well as the ``X channel'' case (with two receivers) where each transmitter has an independent message for each receiver. Finally, it is shown how to generalize this strategy beyond Gaussian channel models. For a class of finite field interference channels, this approach yields the ergodic capacity region.
 \end{abstract}


\section{Introduction}
Consider $K$ transmitter-receiver pairs that communicate over a wireless channel on the same frequency band. If the users are not allowed to cooperate, it is clear that concurrent transmissions will interfere with one another. The key question is at what rate can each pair communicate in the presence of interference from all other pairs. If only one pair is active, this reduces to an interference-free point-to-point communication problem for which the capacity is known. Intuitively, it seems that the best possible scheme for $K$ active pairs would allow each transmitter to operate at roughly $1/K$ its interference-free capacity. Surprisingly, through a new strategy known as \textit{interference alignment} \cite{mmk08, cj08}, it is possible to have each transmitter operate all the way up to $1/2$ its interference-free capacity. The basic idea is that, from the viewpoint of each receiver, the interference should look as if it originated from a single user. For the interference channel, Cadambe and Jafar developed a vector space alignment strategy over many parallel channels (which can be obtained by using multiple frequency bands or time instances). The end result is that each receiver sees its desired signal in half the dimensions while the interfering signals occupy the other half and each user can approach $1/2$ its interference-free capacity as the signal-to-noise ratio (SNR) goes to infinity \cite{cj08}. In this paper, we propose a simple new strategy, \textit{ergodic interference alignment}, that permits each user to achieve at least half its interference-free capacity at any SNR. At its heart, our scheme relies on the availability of time-varying, independent channel coefficients that are drawn from distributions with uniform phase.

We now provide a high-level description of our scheme. Assume that the $K$ transmitters send out signals $X_1, X_2, \ldots, X_K$ at time $t$ under channel matrix $\mathbf{H} = \{h_{k\ell}\}$ and that each receiver observes:
\begin{align}
Y_k[t] = \sum_{\ell = 1}^K{h_{k \ell} X_\ell} + Z_k[t]
\end{align} where $Z_k[t]$ is independent and identically distributed (i.i.d.) additive noise. The transmitters wait until the complementary channel matrix $\Hb_{C}$ occurs at time $t_C$ where 
\begin{align}
\Hb_C = \left[
\begin{array}{cccc}
 h_{11} & -h_{12}  & \cdots& -h_{1K}  \\
 -h_{21} & h_{22}  & \cdots& -h_{2K}  \\
 \vdots & \vdots & \ddots & \vdots \\
 -h_{K1} & -h_{K2} & \cdots & h_{KK}
\end{array}
\right] 
\end{align} and then resend $X_1, X_2, \ldots, X_K$. This gives each receiver access to
\begin{align}
Y_k[t_C] = h_{kk} X_k - \sum_{\ell \neq k}{h_{k\ell} X_\ell} + Z_k[t_C]
\end{align} which it can add to $Y_k[t]$ to get
\begin{align}
Y_k[t] + Y_k[t_C] = 2 h_{kk} X_k + Z_k[t] + Z_k[t_C] \ .
\end{align} So, for the cost of two channel uses, we can get an interference-free channel. The observant reader will have noticed that, for most reasonable fading distributions, any single $\Hb_C \in \Cbb^{K \times K}$ has measure zero and will effectively never occur.  Fortunately, for our purposes, it is enough to wait until the channel matrix is fairly close to $\Hb_C$ to retransmit the signals. The description above is meant only to illustrate the key principles at work and we will make our analysis rigorous in the sequel. 

In some scenarios, each receiver may wish to recover more than one of the transmitted messages. Assume each receiver wants $L$ messages out of the $K > L$ messages that were transmitted. We can think of these messages as unknown variables and allocate one additional unknown variable for the remaining transmitted messages which act as interference. If the transmitters send out the same signals over $L + 1$ appropriately chosen channel matrices, the receivers will have enough ``equations''  to eliminate the interference and solve for their desired messages.  We will generalize our ergodic alignment scheme to this scenario and show that it can also be applied to an X channel with two receivers that each want an independent message from each transmitter. 

In the Gaussian case, each receiver can simply add up its observations from paired channel matrices and then try to recover its desired messages. This is because the desired signal is combined \textit{coherently} while the noise is not which boosts the effective signal-to-noise ratio (SNR). For other channel models, it may be beneficial to remove the noise prior to combining the two observations. We will demonstrate this through the derivation of the capacity region of a finite field interference channel with time-varying channel coefficients. Here, the optimal strategy is to reliably decode equations of the transmitted messages using the computation codes developed in \cite{ng07IT} and then solve for the desired messages. 

\subsection{Related Work}

To date, the capacity region of the Gaussian interference channel is unknown except in some special cases. If the interference strength at each receiver is very strong, then it has been shown that it is optimal to first decode the interference and then extract the desired message\cite{carleial75,sato81,hk81,ssep11}. Conversely, if the interference strength is very weak, it is optimal to treat the interference as noise \cite{mk09,skc09,av09IT}. For the two-user case, Etkin, Tse, and Wang showed that a version of the Han-Kobayashi scheme \cite{hk81} is approximately optimal and achieves the capacity region to within one bit \cite{etw08}.

For interference channels with $K > 2$ transmitter-receiver pairs, interference alignment \cite{mmk08,cj08,js08} offers substantial rate gains. Specifically, Cadambe and Jafar \cite{cj08} showed that $K/2$ degrees-of-freedom are attainable using an alignment scheme that exploits instantaneous channel state information at the transmitters (CSIT), coding across many parallel channels \cite{cj09ITa,ssep11}, and taking a high SNR limit. Subsequent work has focused on developing alignment strategies that can operate outside of this regime. 

One natural question following the results in \cite{cj08} is whether the same gains are attainable at finite SNR. This paper answers this question in the affirmative through a new alignment strategy (under an additional condition on the channel coefficient phases). In concurrent work to our own, \"Ozg\"ur and Tse examined the interference alignment scheme of Cadambe and Jafar \cite{cj08} and found a lower bound on the rate at finite SNR for phase fading \cite{ot09}. In parallel, Jeon and Chung developed an alignment strategy for finite field interference networks \cite{jc09}. For a single-hop interference network, they match up pairs of channel matrices as we do to get interference-free channels. For a multi-hop network, they use subsequent hops to invert the channel matrix from the first hop. This technique was subsequently used to characterize the degrees-of-freedom region for a broad class of layered Gaussian relay networks \cite{jcj11}. Earlier work by Grokop, Tse, and Yates proposed an alignment scheme for line-of-sight interference channels with provably good rates at finite SNR \cite{gty11}. 

Several groups have recently applied the techniques developed here to derive tighter capacity scaling laws for dense wireless networks. Jafar showed that for transmitter-receiver pairs distributed uniformly in the unit square, ergodic alignment yields the exact capacity as the network size goes to infinity \cite{jafar11}. Subsequent work by Aldridge, Johnson, and Piechocki extended this result to a broader class of node placement distributions \cite{ajp09}. Niesen studied multi-hop networks with $K$ nodes with unicast and multicast traffic and found upper and lower bounds that differed by only a $\log K$ factor \cite{niesen11}. For the multiple-access wiretap channel, Bassily and Ulukus have developed a variant of our technique that pairs channel realizations to minimize the information leaked to the eavesdropper \cite{bu10}. 

Another natural question is whether interference alignment is possible over static channels. Bresler, Parekh, and Tse demonstrated that alignment can be achieved on the signal scale using lattice codes and employed this strategy to approximate the capacity of the many-to-one (and one-to-many) interference channel to within a constant number of bits \cite{bpt10}. Lattice-based codes have also been used to characterize a ``very strong'' regime \cite{sjvj08}, the generalized degrees-of-freedom \cite{jv10}, and the approximate sum capacity \cite{oen12} for symmetric $K$-user interference channels. Recent efforts have attempted to generalize this approach to a broader class of channel gains \cite{jv10ITW,oe11}. Motahari \textit{et al.} found that $K/2$ degrees-of-freedom are achievable (up to a set of channel matrices of measure zero) by embedding alignment vectors into scalar irrationals \cite{mgmk09}. However, for rational coefficients, the degrees-of-freedom is strictly less than $K/2$ as shown by Etkin and Ordentlich \cite{eo09}. 

Recent work has also strived to characterize the gains of linear alignment strategies using limited channel realizations. For $3$-user interference channels, Cadambe, Jafar, and Wang showed that linear precoding combined with asymmetric complex signaling offers alignment gains for a single channel realization \cite{cjw10}. Subsequent work by Bresler and Tse found the degrees-of-freedom for symmetric linear alignment for an arbitrary number of channel realizations \cite{bt09}. More recently, several groups have developed feasibility conditions on linear alignment over a single channel realization of a $K$-user MIMO interference channel \cite{bct11,rll11,wgj11B}.

Another interesting line of recent work has developed alignment schemes that do not require instantaneous CSIT. For instance, if the channel coefficients are appropriately correlated, alignment is possible without any CSIT \cite{jafar12}. For independent channel coefficients, alignment is still possible with delayed CSIT \cite{mt10,mjs12,wgj11,vv11,gmk11}, although, in general, the gains are not as high as in the instantaneous case. 

For a more comprehensive overview of the alignment literature, we point to a recent survey \cite{jafar11}.

\subsection{Paper Organization}

The next section provides a formal problem statement for the time-varying interference channel and Section \ref{s:chanquant} develops a quantization scheme that will be useful for our analysis. In Section \ref{s:ergodic}, we show that each receiver can achieve at least half its interference-free rate at any SNR and, in Section \ref{s:delay}, we discuss the delay incurred by this scheme. Section \ref{s:multicast} generalizes ergodic alignment to the case where each receiver wants more than one message. In Section \ref{s:xchannel}, we attempt to extend our scheme to the X channel and give a scheme that works for the $2$-receiver case. All of the prior schemes operate on the symbol level; in Section \ref{s:finitefield}, we show that for non-Gaussian channels, sometimes each receiver should denoise its received signals prior to combining them. Finally, Appendix \ref{s:outer} provides upper bounds for the Gaussian case and Appendix \ref{s:compute} reviews a useful result from computation coding.

\section{Time-Varying Gaussian Interference Channel} \label{s:probstateic}

We begin with some notational conventions. We will denote vectors using boldface lowercase letters and matrices with boldface uppercase letters. Realizations of a random variable are (sometimes) denoted using sans-serif font. For instance, $\mathbb{P}(\mathbf{H} = \mathsf{H})$ denotes the probability that the random matrix $\mathbf{H}$ takes on the value $\mathsf{H}$. All logarithms are to base $2$.

There are $K$ transmitter-receiver pairs that communicate across a narrowband wireless channel over $T$ time steps (see Figure \ref{f:museric}). 

\begin{figure}[h]
\begin{center}
\psset{unit=0.75mm}
\begin{pspicture}(0,-20)(110,45)

\rput(1,32){$m_1$} \psline{->}(5,32)(10,32) \psframe(10,27)(21,37)
\rput(16,32){$\mathcal{E}_1$} \rput(29,35.5){$X_1[t]$}
\psline{->}(21,32)(37,32)

\rput(1,12){$m_2$} \psline{->}(5,12)(10,12) \psframe(10,7)(21,17)
\rput(16,12){$\mathcal{E}_2$} \rput(29,15.5){$X_2[t]$}
\psline{->}(21,12)(37,12)

\rput(16,-1){$\vdots$}

\rput(0.5,-18){$m_K$} \psline{->}(5,-18)(10,-18) \psframe(10,-23)(21,-13)
\rput(16,-18){$\mathcal{E}_K$} \rput(29,-14.5){$X_K[t]$}
\psline{->}(21,-18)(37,-18)


\psframe(37,-23)(65,37)
\rput(51,7){\Large{$\mathbf{H}(t)$}}

\rput(24,0){
\psline(41,32)(47,32)
\pscircle(49,32){2} \psline{-}(48,32)(50,32)
\psline{-}(49,31)(49,33) \psline{->}(49,41)(49,34) \rput(49,44){$Z_1[t]$}
\psline{->}(51,32)(65,32) \rput(57,35.5){$Y_1[t]$}

\psline(41,12)(47,12)
\pscircle(49,12){2} \psline{-}(48,12)(50,12)
\psline{-}(49,11)(49,13) \psline{->}(49,21)(49,14) \rput(49,24){$Z_2[t]$}
\psline{->}(51,12)(65,12) \rput(57,15.5){$Y_2[t]$}

\psline(41,-18)(47,-18)
\pscircle(49,-18){2} \psline{-}(48,-18)(50,-18)
\psline{-}(49,-19)(49,-17) \psline{->}(49,-9)(49,-16) \rput(49.25,-5.75){$Z_K[t]$}
\psline{->}(51,-18)(65,-18) \rput(57,-14.5){$Y_K[t]$}

\psframe(65,27)(77,37) \rput(71.5,32){$\mathcal{D}_1$}
\psline{->}(77,32)(82,32)  \rput(86,32.25){$\hat{m}_1$}

\psframe(65,7)(77,17) \rput(71.5,12){$\mathcal{D}_2$}
\psline{->}(77,12)(82,12) \rput(86,12.25){$\hat{m}_2$}

\rput(71.5,-1){$\vdots$}

\psframe(65,-23)(77,-13) \rput(71.5,-18){$\mathcal{D}_K$}
\psline{->}(77,-18)(82,-18) \rput(87,-17.75){$\hat{m}_K$}

}

\end{pspicture}
\end{center}
\caption{$K$-user Gaussian interference channel with time-varying channel coefficients.} \label{f:museric}
\end{figure}
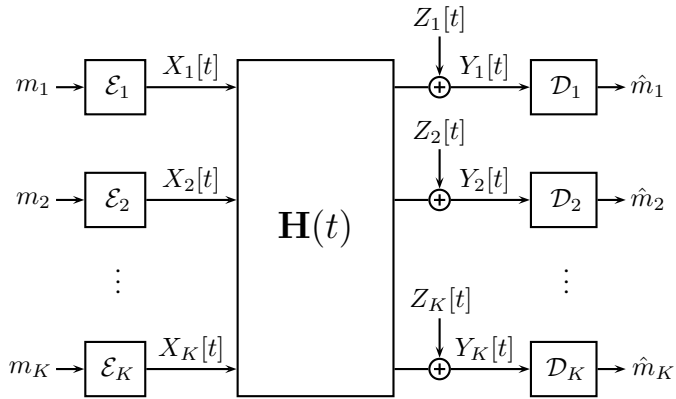

\begin{definition}[Messages] Each transmitter has a \textit{message} $m_\ell$ chosen independently and uniformly from the set $\{1,2,\ldots, 2^{n\tilde{R}_\ell} \}$ for some $\tilde{R}_\ell \geq 0$.
\end{definition}

\begin{definition}[Encoders] Each transmitter has an \textit{encoding function}, $\mathcal{E}_\ell: \{1,2,\ldots,2^{n\tilde{R}_\ell}\} \rightarrow \mathbb{C}^T$, that maps its message $m_\ell$ into a length $T$ channel input $\{X_\ell[t]\}_{t=1}^T$ that satisfies the \textit{power constraint}
\begin{align}
\frac{1}{T} \sum_{t=1}^T{\big|X_\ell[t]\big|^2} \leq P \ .
\end{align}
\end{definition}

\begin{definition}[Channel Model] The channel output observed by each receiver is a noisy linear combination of the inputs
\begin{align}
Y_k[t] = \sum_{\ell=1}^K{h_{k \ell}[t] X_\ell[t]} + Z_k[t]
\end{align} where the $h_{k\ell}[t]$ are time-varying channel coefficients and $Z_k[t]$ is additive i.i.d. noise and drawn from a circularly symmetric complex Gaussian distribution with unit variance, $Z_k[t] \sim \mathcal{CN}(0, 1)$. Let $\mathbf{H}[t] = \left\{h_{k\ell}[t]\right\}$ denote the matrix of channel coefficients at time $t$. Each entry of this matrix is independent of the others for all $t$ and the channel matrix itself is i.i.d. across time,
\begin{align}
f_{\Hb}(\Hs) &= \prod_{k=1}^K \prod_{\ell=1}^K f_{h_{k\ell}}(\mathsf{h}_{k\ell}) \\
f_{\Hb[1]\cdots\Hb[T]}(\mathsf{H}_1, \ldots, \mathsf{H}_T) &= \prod_{t=1}^T f_{\Hb}(\Hs_t) \ . 
\end{align} We assume that the phase of each channel coefficient is drawn according to a uniform distribution and independent from its magnitude, 
\begin{align}
f_{h_{k\ell}}(\mathsf{h}) = f_{h_{k\ell}}( e^{jb} \mathsf{h})~~~\forall \mathsf{h} \in \mathbb{C},b \in [0, 2\pi)\ .
\end{align} 

\begin{remark}
Although our alignment scheme requires the phases to be drawn from uniform distributions, this requirement can be relaxed by changing how channel matrices are paired. See \cite{gj12} for a recent study of ergodic alignment under asymmetric phase distributions.   
\end{remark}

The transmitted symbols at time $t$ can depend on the channel realizations up to and including time $t$. This is the usual notion of causal CSIT. Let $\mathbf{H}[t] = \left\{h_{k\ell}[t]\right\}$ denote the matrix of channel coefficients at time $t$.
\end{definition}

\begin{remark} Channel coefficients that change at every time step are often referred to as a fast fading process. For our considerations, we just need that there is sufficient variation of the channel coefficients over the duration of a codeword. The assumption that the channel coefficients are i.i.d. across time is taken to simplify the analysis. 
\end{remark}

\begin{remark}  We can model the effect of different power constraints at each transmitter and different noise variances at each receiver by modifying the coefficient probability distributions.
\end{remark}

\begin{definition}[Decoders]\label{d:decoders} Each receiver has a \textit{decoding function}, $\mathcal{D}_k: \mathbb{C}^T \rightarrow \{1,2,\ldots,2^{n\tilde{R}_k}\}$, that maps its length $T$ observed channel output $\{Y_k[t]\}_{t=1}^T$ into an estimate $\widehat{m}_k$ of its desired message $m_k$. 
\end{definition}

\begin{definition}[Achievable Rates]\label{d:rates} We say that a rate tuple $(R_1,R_2,\ldots,R_K)$ is \textit{achievable} if for all $\epsilon > 0$ and $n$ large enough there exist channel encoding and decoding functions $\mathcal{E}_1, \ldots, \mathcal{E}_K, \mathcal{D}_1, \ldots, \mathcal{D}_K$ such that
\begin{align}
&\tilde{R}_k > R_k - \epsilon,~~~k=1,2,\ldots,K, \\
&\P\Big(\{\widehat{m}_1 \neq m_1\} \cup \ldots \cup \{\widehat{m}_K \neq m_K\} \Big) < \epsilon \ .
\end{align}
\end{definition}

\begin{definition}[Capacity] The \textit{capacity region} is the closure of the set of all achievable rate tuples. \end{definition}

\section{Channel Quantization} \label{s:chanquant}
Our scheme relies on matching up channel matrices so that the interference terms cancel out when we sum up the matrices. Clearly, given any channel matrix $\mathbf{H}$, the probability that its exact complement $\mathbf{H}_C$ will occur is zero (for continuous-valued fading). Thus, we can only match up matrices \textit{approximately}. We will accomplish this by quantizing the channel coefficients and matching up matrices based on their quantized values. By taking finer and finer quantizations, we can achieve the target rate in the limit. 

We also need to ensure that nearly all matrices that occur will be successfully paired up with their complements. Since the coefficients are drawn i.i.d. from distributions with uniform phase, the probability that the complement of a channel matrix occurs in a given time step is the same as the probability that the original matrix occurs. We will constrain the quantized matrices to lie within a finite set by throwing out any matrices with coefficients larger than a threshold. Finally, we choose the blocklength to be large enough so that the sequence of quantized channel matrices is strongly typical with high probability. This means that the empirical distribution of channel matrices will be close to the true distribution which implies that nearly all matrices can be matched.

Let $\hmax$ denote the channel coefficient threshold. We will ignore any channel matrix that contains at least one coefficient with magnitude larger than $\hmax$. Let 
\begin{align}
\mathcal{L} \triangleq \left\{ \mathbf{H} \in \mathbb{C}^{K \times K}: |h_{k\ell}| > \hmax~~\mbox{for some }k,\ell \right\}
\end{align} denote the set of all matrices that violate the threshold and let 
\begin{align}
\rho \triangleq \P\big(\Hb[t] \in \mathcal{L}\big)
\end{align} be the probability of some matrix in this set occurring at time $t$. Note that $\rho$ is a decreasing function of $\hmax$.

We now define the quantization function $q$ for the channel coefficients. The complex plane up to distance $h_{\text{MAX}}$ from the origin is divided up into $\kappa$ disjoint rings of equal width. These rings are further subdivided into equal segments based on $\eta$ angles spaced equally between $0$ and $2\pi$. The parameters $\kappa$ and $\eta$ are chosen to be large enough such that the maximum distance between any two points within a segment is $\delta$ where $\delta > 0$ will be specified later. Each segment is a quantization cell for the channel coefficients which we represent by its centroid. Thus, $q\big(h_{k\ell}[t])$ maps $h_{k\ell}[t]$ to the centroid within its segment if $\big| h_{k\ell}[t]\big| \leq \hmax$. If $\big| h_{k\ell}[t]\big| > \hmax$, then $q\big(h_{k\ell}[t])$ maps to an erasure symbol $\Gamma$. See Figure \ref{f:quant} for an illustration of this quantization scheme. 

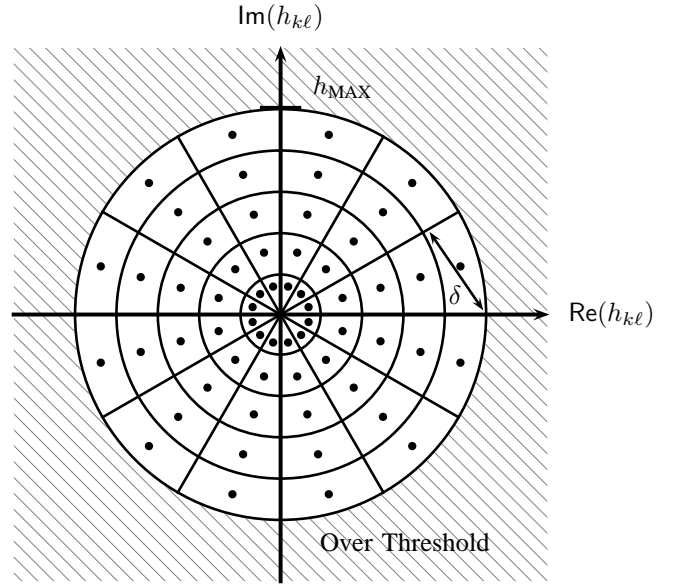
\begin{figure}[h]
\begin{center}
\psset{unit=0.55mm}
\begin{pspicture}(-65,-65)(85,80)

\psframe[hatchwidth=0.3pt,hatchcolor=gray,fillstyle=vlines,linecolor=white](-65,-65)(65,65)

\pscircle[fillstyle=solid,fillcolor=white,linecolor=white](0,0){50}

\psline[linewidth=1.5pt]{->}(-65,0)(65,0)
\psline[linewidth=1.5pt]{->}(0,-65)(0,65)
\rput(80,0){$\mathsf{Re}(h_{k\ell})$}
\rput(0,72){$\mathsf{Im}(h_{k\ell})$}

\rput{0}(0,0){
\rput{30}(0,0){
\psline[linewidth=1pt](-50,0)(50,0)
}
\rput{60}(0,0){
\psline[linewidth=1pt](-50,0)(50,0)
}
\rput{90}(0,0){
\psline[linewidth=1pt](-50,0)(50,0)
}
\rput{120}(0,0){
\psline[linewidth=1pt](-50,0)(50,0)
}
\rput{150}(0,0){
\psline[linewidth=1pt](-50,0)(50,0)
}
\rput{180}(0,0){
\psline[linewidth=1pt](-50,0)(50,0)
}
}

\rput{15}(0,0){
\pscircle[fillstyle=solid,fillcolor=black](7,0){1}
\pscircle[fillstyle=solid,fillcolor=black](15.5,0){1}
\pscircle[fillstyle=solid,fillcolor=black](25,0){1}
\pscircle[fillstyle=solid,fillcolor=black](35,0){1}
\pscircle[fillstyle=solid,fillcolor=black](45,0){1}
}
\rput{45}(0,0){
\pscircle[fillstyle=solid,fillcolor=black](7,0){1}
\pscircle[fillstyle=solid,fillcolor=black](15.5,0){1}
\pscircle[fillstyle=solid,fillcolor=black](25,0){1}
\pscircle[fillstyle=solid,fillcolor=black](35,0){1}
\pscircle[fillstyle=solid,fillcolor=black](45,0){1}}
\rput{75}(0,0){
\pscircle[fillstyle=solid,fillcolor=black](7,0){1}
\pscircle[fillstyle=solid,fillcolor=black](15.5,0){1}
\pscircle[fillstyle=solid,fillcolor=black](25,0){1}
\pscircle[fillstyle=solid,fillcolor=black](35,0){1}
\pscircle[fillstyle=solid,fillcolor=black](45,0){1}}
\rput{105}(0,0){
\pscircle[fillstyle=solid,fillcolor=black](7,0){1}
\pscircle[fillstyle=solid,fillcolor=black](15.5,0){1}
\pscircle[fillstyle=solid,fillcolor=black](25,0){1}
\pscircle[fillstyle=solid,fillcolor=black](35,0){1}
\pscircle[fillstyle=solid,fillcolor=black](45,0){1}}
\rput{135}(0,0){
\pscircle[fillstyle=solid,fillcolor=black](7,0){1}
\pscircle[fillstyle=solid,fillcolor=black](15.5,0){1}
\pscircle[fillstyle=solid,fillcolor=black](25,0){1}
\pscircle[fillstyle=solid,fillcolor=black](35,0){1}
\pscircle[fillstyle=solid,fillcolor=black](45,0){1}}
\rput{165}(0,0){
\pscircle[fillstyle=solid,fillcolor=black](7,0){1}
\pscircle[fillstyle=solid,fillcolor=black](15.5,0){1}
\pscircle[fillstyle=solid,fillcolor=black](25,0){1}
\pscircle[fillstyle=solid,fillcolor=black](35,0){1}
\pscircle[fillstyle=solid,fillcolor=black](45,0){1}}
\rput{195}(0,0){
\pscircle[fillstyle=solid,fillcolor=black](7,0){1}
\pscircle[fillstyle=solid,fillcolor=black](15.5,0){1}
\pscircle[fillstyle=solid,fillcolor=black](25,0){1}
\pscircle[fillstyle=solid,fillcolor=black](35,0){1}
\pscircle[fillstyle=solid,fillcolor=black](45,0){1}}
\rput{225}(0,0){
\pscircle[fillstyle=solid,fillcolor=black](7,0){1}
\pscircle[fillstyle=solid,fillcolor=black](15.5,0){1}
\pscircle[fillstyle=solid,fillcolor=black](25,0){1}
\pscircle[fillstyle=solid,fillcolor=black](35,0){1}
\pscircle[fillstyle=solid,fillcolor=black](45,0){1}}
\rput{255}(0,0){
\pscircle[fillstyle=solid,fillcolor=black](7,0){1}
\pscircle[fillstyle=solid,fillcolor=black](15.5,0){1}
\pscircle[fillstyle=solid,fillcolor=black](25,0){1}
\pscircle[fillstyle=solid,fillcolor=black](35,0){1}
\pscircle[fillstyle=solid,fillcolor=black](45,0){1}}
\rput{285}(0,0){
\pscircle[fillstyle=solid,fillcolor=black](7,0){1}
\pscircle[fillstyle=solid,fillcolor=black](15.5,0){1}
\pscircle[fillstyle=solid,fillcolor=black](25,0){1}
\pscircle[fillstyle=solid,fillcolor=black](35,0){1}
\pscircle[fillstyle=solid,fillcolor=black](45,0){1}}
\rput{315}(0,0){
\pscircle[fillstyle=solid,fillcolor=black](7,0){1}
\pscircle[fillstyle=solid,fillcolor=black](15.5,0){1}
\pscircle[fillstyle=solid,fillcolor=black](25,0){1}
\pscircle[fillstyle=solid,fillcolor=black](35,0){1}
\pscircle[fillstyle=solid,fillcolor=black](45,0){1}}
\rput{345}(0,0){
\pscircle[fillstyle=solid,fillcolor=black](7,0){1}
\pscircle[fillstyle=solid,fillcolor=black](15.5,0){1}
\pscircle[fillstyle=solid,fillcolor=black](25,0){1}
\pscircle[fillstyle=solid,fillcolor=black](35,0){1}
\pscircle[fillstyle=solid,fillcolor=black](45,0){1}}

\psline[linewidth=1.5pt](-5,50)(5,50)
\rput(15,55){$h_{\text{MAX}}$}

\psline[linewidth=1.0pt]{<->}(36,20)(49,1)
\rput(42.15,4.85){$\delta$}

\pscircle[linewidth=1pt](0,0){10}
\pscircle[linewidth=1pt](0,0){20}
\pscircle[linewidth=1pt](0,0){30}
\pscircle[linewidth=1pt](0,0){40}
\pscircle[linewidth=1pt](0,0){50}

\rput(30,-55){Over Threshold}

\end{pspicture}
\end{center}
\caption{Quantizing complex-valued channel coefficients $h_{k\ell}$ with magnitude less  than $h_{\text{MAX}}$ to a finite set. Here, the number of rings is $\kappa = 5$ and the number of segments per ring is $\eta = 12$. The maximum distance between any two points in a quantization cell is $\delta$.} \label{f:quant}
\end{figure}

Throughout the paper, we will match up channel coefficients based on their quantization cells. For notational convenience, let
\begin{align}
\hh_{k\ell}[t] \triangleq q\big(h_{k\ell}[t]\big)
\end{align} denote the quantized channel coefficients.

One important aspect of this quantization scheme is that each segment has the same probability of occurring as any other segment within the same ring. Note that this depends strongly on the assumptions of uniform phase and the independence of phase and magnitude. Ideally, we would pair up channel coefficients to cancel out the interference exactly. However, as explained above, this is not possible at any finite blocklength. The following lemma bounds the effect of combining channel coefficients based on their quantization cells.

\begin{lemma}\label{l:quantbounds}
Let $h_{k\ell}[t_1],h_{k\ell}[t_2],\ldots,h_{k\ell}[t_N]$ be channel coefficients with magnitudes less than $h_{\text{MAX}}$. Then, for any $a_n \in \Cbb$, \begin{align}
\left|\sum_{n=1}^N a_n h_{k\ell}[t_n] \right| &\leq \left| \sum_{n=1}^N a_n \hat{h}_{k\ell}[t_n] \right| + \delta \sum_{n=1}^N |a_n| \\
\left|\sum_{n=1}^N a_n h_{k\ell}[t_n] \right| &\geq \max\left( 0,~\left| \sum_{n=1}^N a_n \hat{h}_{k\ell}[t_n] \right| - \delta \sum_{n=1}^N |a_n|\right) \nonumber
\end{align} where $\delta$ is the maximum distance between any two points in a quantization cell.
\end{lemma}
\begin{IEEEproof}
Define $e_{k\ell}[t_n] \triangleq h_{k\ell}[t_n] - \hat{h}_{k\ell}[t_n]$. Since the coefficient magnitudes are less than $h_{\text{MAX}}$, then $| e_{k\ell}[t_n] | < \delta$. By the triangle inequality,
\begin{align}
\left| \sum_{n=1}^N{a_n h_{k\ell}[t_n]}\right| &=  \left| \sum_{n=1}^N{a_n \big(\hat{h}_{k\ell}[t_n] + e_{k\ell}[t_n]\big)} \right|\\
&\leq  \left| \sum_{n=1}^N{a_n \hat{h}_{k\ell}[t_n]} \right| +  \delta \sum_{n=1}^N{|a_n|}\ .
\end{align} The second inequality follows similarly via the reverse triangle inequality. \end{IEEEproof}

Channel matrices are quantized simply by quantizing their individual coefficients
\begin{align}
\mathbf{\hat{H}}[t] \triangleq \big\{ \hat{h}_{k\ell}[t] \big\} \ .
\end{align} Let $\mathcal{H}$ denote the finite set onto which channel matrices are quantized.

To facilitate our analysis, we will split the $T$ time slots into $N$ consecutive blocks of $T/N$ time slots each. Let 
\begin{align}
\mathbf{H}^{(n)} \triangleq \Bigg( \mathbf{H}\bigg[1+ \frac{(n-1)T}{N} \bigg],~\ldots~,\mathbf{H}\bigg[\frac{nT}{N} \bigg]\Bigg)
\end{align} for $n \in \{1,2,\ldots,N\}$ and let $\mathbf{\hat{H}}^{(n)}$ denote the corresponding quantized sequence.

We now recall the notion of strong typicality for sequences of discrete random variables and specialize it to sequences of quantized channel matrices. We define \begin{align}
p_{\mathbf{\hat{H}}}\big(\mathsf{\hat{H}}\big) \triangleq \P\big(\mathbf{\hat{H}}[t] = \mathsf{\hat{H}}\big)
\end{align} to be the probability under the fading distribution that a channel matrix quantizes to $\mathsf{\hat{H}} \in \mathcal{H}$. Also, define
\begin{align}
\#\Big(\mathsf{\hat{H}}|\mathbf{\hat{H}}^{(n)}\Big) \triangleq \bigg|\bigg\{ t : \mathbf{\hat{H}}[t] = \mathsf{\hat{H}},~1+ \frac{(n-1)T}{N}  \leq t \leq \frac{nT}{N} \bigg\}\bigg| \nonumber
\end{align} to be the number of quantized channel matrices within the $n^{\text{th}}$ block that are equal to $\mathsf{\hat{H}} \in \mathcal{H}$.

\begin{definition}[Strong Typicality] \label{d:typical}A block of quantized channel matrices, $\mathbf{\hat{H}}^{(n)}$, is \textit{$\gamma$-typical} if
\begin{align}
\bigg| \frac{N}{T}\#\Big(\mathsf{\hat{H}}|\mathbf{\hat{H}}^{(n)}\Big)-p_{\mathbf{\hat{H}}}\big(\mathsf{\hat{H}}\big) \bigg| \leq \gamma ~~~ \forall \mathsf{\hat{H}} \in \mathcal{H} \ . \end{align}
\end{definition} 

\begin{lemma} \label{l:typical}
For any $\epsilon > 0$ and $T$ large enough, the probability that all blocks $\Hhb^{(1)},\ldots, \Hhb^{(N)}$ are $\gamma$-typical is lower bounded by $1-\epsilon$.
\end{lemma} 
\begin{IEEEproof}
From Lemma 2.12 in \cite{csiszarkorner}, the probability that a block $\Hhb^{(n)}$ is $\gamma$-typical is at least
\begin{align}
1 - \frac{|\mathcal{H}|N}{4T\gamma^2} \ .
\end{align} Since the blocks are independent, the probability that all blocks $\Hhb^{(1)},\ldots,\Hhb^{(N)}$ are $\gamma$-typical is lower bounded by
\begin{align}
\bigg(1 - \frac{|\mathcal{H}|N}{4T\gamma^2}\bigg)^N \ . \label{e:probtypical}
\end{align}
From our choice of quantization scheme, $| \mathcal{H}| = (\kappa \eta + 1)^{K^2}$. Thus, (\ref{e:probtypical}) goes to $1$ as $T$ goes to infinity which completes the proof. 
\end{IEEEproof}
We will only work with sequences of channel matrices that are $\gamma$-typical and declare errors on the rest. This ensures that nearly all time indices can be matched up appropriately.

\section{Ergodic Interference Alignment} \label{s:ergodic}

Each transmitter-receiver pair would clearly be better off if it had exclusive access to the channel and faced no interference from other users. Specifically, if $h_{k\ell} = 0~~\forall \ell \neq k$, each receiver sees a point-to-point channel from its transmitter and can achieve
\begin{align}
R_k = \E\big[\log{\left(1 +  |h_{kk}|^2 P \right)}\big]\ .
\end{align}We call this the \textit{interference-free rate} and will use it as a benchmark to gauge our performance.

\begin{remark} Note that this assumes a uniform power allocation across all time slots and one can do better by using the causal channel state information to optimize the power allocation \cite{cs99}. For simplicity, we use a uniform power allocation throughout our derivations. See \cite{tuninetti08} for a study of power allocation for fast fading $2$-user interference channels. The interplay of interference alignment and waterfilling is an interesting subject for future study. 
\end{remark}

A simple approach to interference management is to have transmitters take turns using the channel, often referred to as time-division. For instance, if we partition the channel equally between transmitters, each one can achieve
\begin{align}
R_k =\frac{1}{K}\E\big[\log{(1 +  K|h_{kk}|^2 {P} )}\big] \ . \label{e:tdma}
\end{align} The extra $K$ factor inside the logarithm comes from saving up power while the transmitter is required to stay silent. Under this approach, the sum rate stays nearly constant as we add users to the network. Although this seems like a fundamental performance barrier, we can in fact do much better using interference alignment. 

The main idea underlying alignment is to carefully design the transmission scheme so that the effective interference at each receiver appears as if it came from a single transmitter. For the channel model under consideration, Cadambe and Jafar showed that this is possible using a vector space strategy \cite{cj08}. In brief, their strategy groups together several channel uses to get a (virtual) multiple-input multiple-output (MIMO) interference channel. Each transmitter is assigned a linear transformation based on the fading realization with rank roughly equal to half the number of channel uses. At each receiver, the interfering signals occupy one half of the dimensions while the desired signal occupies the other half and can be extracted using zero-forcing. This strategy allows the sum rate to increase linearly with the number of users at high SNR. Recall that $f(x) = o(g(x))$ means that $\lim_{x \rightarrow \infty}{f(x)/g(x)} = 0$.
\begin{theorem}[Cadambe-Jafar]
For the time-varying Gaussian interference channel, each transmitter can achieve a rate satisfying
\begin{align}
R_k =\frac{1}{2} \log{\left(1 + P\right)} + o(\log{\left(1 + P\right)}) \ .
\end{align} 
\end{theorem} For a full proof, see \cite[Theorem 1]{cj08}. This result characterizes the ``pre-log'' term of the achievable rates (also referred to as the degrees-of-freedom)\footnote{Note that this high SNR result does not depend on the uniform phase assumption.}. It implies that, at sufficiently high SNR, each user can achieve one half its interference-free rate regardless of the number of users in the network. We now set out to prove that each user can achieve at least half its interference-free rate at any SNR.

\begin{theorem}\label{t:icalign}
For the time-varying Gaussian interference channel defined in Section \ref{s:probstateic}, the rates \begin{align}
R_{k} = \frac{1}{2}\E\big[\log{(1 + 2 |h_{kk}|^2 {P})}\big]
\end{align} are achievable for $k = 1,2,\ldots,K$. \end{theorem}
\begin{IEEEproof}
Choose $\epsilon > 0$. We will divide up the $T$ channel uses into two consecutive intervals of equal length. We quantize all channel realizations using the scheme described in Section \ref{s:chanquant}. Applying Lemma \ref{l:typical} with $N = 2$, it follows that both blocks of $T/2$ channel uses are $\gamma$-typical with probability at least $(1 - \frac{\epsilon}{2})$ (with $\gamma$ to be specified later). By Definition \ref{d:typical}, this means that the number of occurrences of each possible quantized channel matrix in each interval is bounded as follows:
\begin{align}
\frac{T}{2}\Big(p_{\Hhb}\big(\Hhs\big) - \gamma \Big) \leq \#\big(\Hhs | \Hhb^{(n)}\big) \leq \frac{T}{2}\Big(p_{\Hhb}\big(\Hhs\big) + \gamma \Big) 
\end{align} for all $\Hhs \in \mathcal{H}$. 

If either interval is not $\gamma$-typical, we declare an error. Otherwise, we know that each quantized matrix will occur at least $\frac{T}{2}\left(p_{\Hhb}\big(\Hhs\big) - \gamma \right)$ times in each interval. A time slot $t$ in an interval is useable unless:
\begin{enumerate}
\item The channel matrix $\mathbf{H}[t]$ contains one or more elements with magnitude larger than $h_{\text{MAX}}$. 
\item The channel matrix $\mathbf{H}[t]$ does not violate the threshold but the corresponding quantized matrix $\Hhb[t]$ has already occurred at least  $\frac{T}{2}\big(p_{\Hhb}\big(\Hhs\big) - \gamma \big)$ times. 
\end{enumerate}
Assuming all intervals are $\gamma$-typical, the number of useable time slots per interval is \begin{align}
 &\left\lfloor \frac{T}{2}\sum_{\Hhs: \hat{h}_{k\ell} \neq \Gamma }\Big(p_{\Hhb}\big(\Hhs\big) - \gamma \Big) \right\rfloor= \left\lfloor \frac{T}{2}\left(1 - \rho - (\kappa \eta)^{K^2}\gamma\right)\right\rfloor.\nonumber
\end{align} Recall that $\rho$ is the probability the channel matrix contains an element larger than $h_{\text{MAX}}$ (which corresponds to the quantized matrix containing an erasure symbol $\Gamma$) and $\kappa$ and $\eta$ are parameters in the channel quantization. 

Each encoder uses an independent codebook $\mathcal{C}_k$ with rate $\tilde{R}_k$ and length equal to the number of useable time slots per interval. Each codebook is generated elementwise i.i.d. from a circularly symmetric Gaussian distribution with variance slightly less than $P$ (to ensure that, for large blocklengths, the power constraint is satisfied).

During the first interval, each transmitter sends out a new symbol from its codeword during each useable time slot $t_1$ and records the corresponding quantized channel matrices $\Hhb[t_1]$. We match up each useable time slot $t_1$ from the first interval with a useable time slot $t_2$ from the second interval for which the quantized channel matrix $\Hhb[t_2]$ is \textit{complementary},
\begin{align}
\Hhb[t_2]=\left[
\begin{array}{cccc}
\hat{h}_{11}[t_1] & -\hat{h}_{12}[t_1]  & \cdots& -\hat{h}_{1K}[t_1]  \\
 -\hat{h}_{21}[t_1] & \hat{h}_{22}[t_1]  & \cdots& -\hat{h}_{2K}[t_1]  \\
 \vdots & \vdots & \ddots & \vdots \\
 -\hat{h}_{K1}[t_1] & -\hat{h}_{K2}[t_1] & \cdots & \hat{h}_{KK}[t_1]
\end{array}
\right] \ .\nonumber
\end{align} Note that this can be done using only causal channel knowledge by greedily matching up time slots from the first interval in the order in which they occur. To ensure that $-\hat{h}_{k\ell}$ corresponds to a valid quantization cell, we constrain the number of angles (given by $\eta$) to be even. Since each channel coefficient has uniform phase, all of the useable time slots from the first interval can be matched with useable time slots from the second interval (assuming that the intervals are $\gamma$-typical).

Each receiver adds up its observations from the first interval to the matched observations in the second time slot $Y_k[t_1] + Y_k[t_2]$. We now calculate the resulting signal-to-interference-and-noise ratio ($\mathsf{SINR}$) at each receiver.  The channel coefficient corresponding to the desired signal belongs to the same quantization cell in matched time slots, $\hat{h}_{kk}[t_1] = \hat{h}_{kk}[t_2]$. From Lemma \ref{l:quantbounds}, we have that
\begin{align}
\big|h_{kk}[t_1] + h_{kk}[t_2]\big| &\geq \max\Big(2 \big|\hat{h}_{kk}[t_1]\big| - 2 \delta,0\Big) \\
&\geq  \max\Big(2 \big|{h}_{kk}[t_1]\big| - 4 \delta,0\Big)
\end{align} where $\delta$ is the maximum distance between any two points in a quantization cell. It follows that the signal power in $Y_k[t_1] + Y_k[t_2]$ for the symbol $X_k[t_1]$ is at least
\begin{align}
(2 \big|h_{kk}[t_1]\big| - 4 \delta)^2 P
\end{align} if $\big|h_{kk}[t_1]\big| > 2\delta$. For interfering signals, the channel coefficients from matched time slots satisfy $\hat{h}_{k\ell}[t_1] = -\hat{h}_{k\ell}[t_2]$. Applying Lemma \ref{l:quantbounds}, we get that 
\begin{align}
\big|h_{kk}[t_1] + h_{kk}[t_2]\big| \leq 2\delta \ .
\end{align} It follows that the total interference power in $Y_k[t_1] + Y_k[t_2]$ is at most
\begin{align}
4 \delta^2 (K-1) P \ .
\end{align} The noise power in $Y_k[t_1] + Y_k[t_2]$ is $2$. Combining these bounds, we get that if $h_{kk}[t_1] > \delta$, the $\mathsf{SINR}$ at receiver $k$ is at least
\begin{align}
\mathsf{SINR}_k \geq  \frac{P \left(2|h_{kk}[t_1]| -  4\delta  \right)^2}{4\delta^2(K-1)P + 2} \label{e:sinrlower} \ .
\end{align} Taking $\delta \rightarrow 0$, we see that
\begin{align}
\lim_{\delta \rightarrow 0} \mathsf{SINR}_k \geq 2|h_{kk}[t_1]|^2 P \ .
\end{align}

By choosing $h_{\text{MAX}}$ large enough, we can make the probability $\tau$ that the channel matrix violates the threshold as small as we desire. Next, we can choose $\kappa$ (the number of quantization rings) and $\eta$ (the number of angles) large enough, to make $\delta$ as small as desired and get the $\mathsf{SINR}$ at each receiver to be as close to $2|h_{kk}[t_1]|^2 P$ as we would like. Then, we can choose $\gamma$ to be sufficiently small so that the fraction of useable time slots is large. Finally, by taking $T$ large enough, we can find a good code with probability of error at most $\frac{\epsilon}{2}$ and rate at least \begin{align}
\frac{1}{2}\E\left[\log{\left(1 + 2 |h_{kk}|^2P\right)}\right] - \frac{\epsilon}{2} .
\end{align} Recall also that with probability $\frac{\epsilon}{2}$ the channel is not $\gamma$-typical so the total probability of error is less than $\epsilon$. Thus, there must exist a set of good fixed codebooks with the same performance. Finally, we expurgate all codewords that violate the power constraint which results in a rate loss of at most $\epsilon/2$ for $T$ large enough.
\end{IEEEproof}

In Figure \ref{f:icplot}, we have plotted the performance of the ergodic alignment scheme from Theorem \ref{t:icalign} for a time-varying $10$-user interference channel with i.i.d. Rayleigh fading, $h_{k\ell} \sim \mathcal{CN}(0,1)$. For comparison, we have also plotted the upper bound from (\ref{e:icupper}) in Appendix \ref{s:outer} and the performance of time division from (\ref{e:tdma}). For all three curves, we have taken a uniform power allocation across time. Note that the rates for ergodic alignment and the upper bound only depend on the fading statistics, not the number of users.

\begin{figure}[h]
\centering
\includegraphics[width=3.8in]{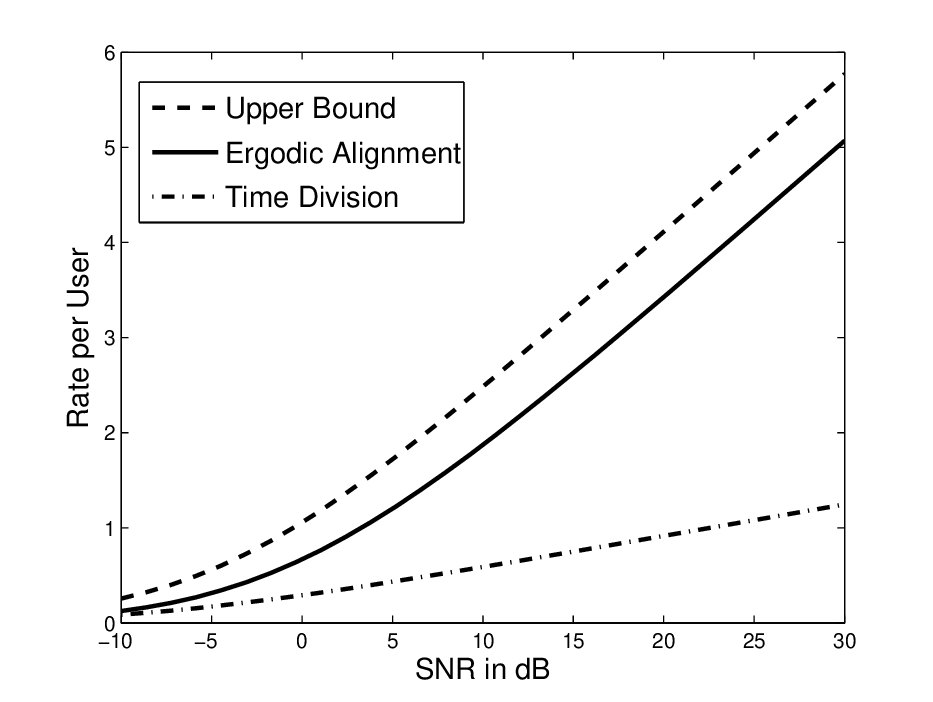}
\caption{Rate per user for the ergodic alignment scheme over a time-varying interference channel with i.i.d. Rayleigh fading. For comparison, we have also plotted the performance of time division for $10$ users.}\label{f:icplot}
\end{figure}

\begin{remark}
Assume all the channel coefficients have equal magnitudes and random, uniform phases, $h_{k\ell} = \exp(j 2\pi \phi_{k\ell})$ for $\phi_{k\ell}\sim \text{Unif}[0,2\pi)$. A quick comparison of the upper bound in (\ref{e:icupper}) and Theorem \ref{t:icalign} reveals that ergodic alignment achieves the sum capacity. Jafar \cite{jafar11IT} has shown that this holds more generally through the concept of a ``bottleneck state.'' That is, if each receiver sees an interferer of equal strength to their desired signal, ergodic alignment is optimal.
\end{remark}

In general, ergodic alignment alone does not yield the capacity region. For instance, if the cross-channel gains are very small relative to the direct gains, then it is better to treat the interference as noise, rather than spending two channel uses to cancel it out \cite{mk09,skc09,av09IT}. Thus, for Rayleigh fading, we can achieve higher rates by using this weak interference strategy over certain channel matrices and the alignment strategy over the rest. Conversely, if the cross-channel gains are very large relative to the direct gains, then it is better to decode the interference prior to decoding the desired message \cite{carleial75,sato81,hk81,ssep11}. It remains unclear as to whether an appropriate mixture of these three schemes can be used to approach the ergodic capacity region. That said, as the number of users increases, it becomes more likely that the network will be in a ``bottleneck state,'' implying that ergodic alignment is optimal \cite{jafar11IT}.

Suppose that user $\ell$ wants to communicate at more than half its interference-free rate. We now propose a simple time-sharing strategy for this scenario that blends our alignment scheme with a time-division scheme.  

\begin{corollary}
For the time-varying Gaussian interference channel defined in Section \ref{s:probstateic}, the following rates are achievable \begin{align}
R_\ell &= \alpha \E\big[\log{(1 + |h_{\ell\ell}|^2P)}\big]\nonumber \\&\qquad+ \frac{(1-\alpha)}{2} \E\left[\log{\left(1 + 2|h_{\ell\ell}|^2P\right)}\right] \\
R_k &= \frac{(1-\alpha)}{2} \E\bigg[\log{\bigg(1 + \frac{2|h_{kk}|^2P}{1-\alpha}\bigg)}\bigg]~~~~k \neq \ell
\end{align} for any $0 \leq \alpha \leq 1$ if each channel coefficient is drawn from a distribution with uniform phase.
\end{corollary}
\begin{IEEEproof}
For $\alpha T$ channel uses, all users except $\ell$ are silent. User $\ell$ employs a standard point-to-point channel code to achieve rate $\E[\log{(1 + |h_{\ell\ell}|^2P)}]$ over these channel uses. For the remaining $(1-\alpha)T$ channel uses, we employ ergodic interference alignment as in the proof of Theorem \ref{t:icalign}. User $\ell$ achieves $\frac{1}{2}\E\left[\log{\left(1 + 2|h_{\ell\ell}|^2P\right)}\right]$ over these channel uses as before. Since each user $k \neq \ell$ was silent for the prior $\alpha T$ channel uses, it has saved up power and can afford to transmit each symbol with average power $\frac{P}{1-\alpha}$ resulting in a rate of $\frac{1}{2}\E\left[\log{\left(1 + \frac{2|h_{kk}|^2P}{1-\alpha}\right)}\right]$ over the remaining channel uses. 
\end{IEEEproof}

\subsection{Delay Analysis} \label{s:delay}

We now provide a brief analysis of the delay requirements of our scheme. First, we note that we designed our matrix pairing strategy to simplify the achievability proof; there may be other choices that will result in lower delay. In general, the delay incurred by an alignment scheme will depend on the number on users $K$, the SNR, and the achievable rates. 

For our analysis, we consider the special case of fixed-magnitude channel gains. Let $h_{k\ell}[t] = r_{k\ell} e^{j \phi_{k\ell}[t]}$ where the magnitudes $r_{k\ell}$ are real, positive constants and the phases $\phi_{k\ell}[t]$ are i.i.d. according to a uniform distribution over $[0,2\pi)$. Since the magnitude is held constant, our quantization scheme from Section \ref{s:chanquant} consists of mapping the phases $\phi_{k\ell}[t]$ to the closest of the $\eta$ quantized angles. The maximum distance $\delta_{k\ell}$ between two channel gains $h_{k\ell}[t_1]$ and $h_{k\ell}[t_2]$ that are quantized to the same angle is
\begin{align}
\delta_{k\ell} &= 2 r_{k\ell} \sin\left(\frac{\pi}{\eta}\right) \\
&\leq \frac{2\pi r_{k\ell}}{\eta} \ .
\end{align} Plugging this into \eqref{e:sinrlower}, we find that the SINR per codeword symbol is lower bounded by 
\begin{align}
\mathsf{SINR}_k \geq  \frac{r_{kk}^2 P \left( 2 - \frac{8\pi}{\eta}  \right)^2}{\frac{16\pi^2}{\eta^2}(K-1)P \max_{\ell \neq k} r_{k\ell}^2  + 2}  \ .
\end{align} To maintain a capacity scaling of roughly $\frac{1}{2} \log P$ per user, we require that, for some constant $\phi > 0$, 
\begin{align}
\eta^2 = (K-1) P \phi \ .
\end{align} Consider the expected delay before a single codeword symbol from each transmitter is successfully obtained by the receivers. If the transmitters were to send a new symbol every time slot and the receivers were to simply treat interference as noise, this expected delay is 1. For ergodic alignment, codeword symbols must travel through a channel matrix and the complementary channel matrix with opposite quantized phases. Since the channel gains are independent, the probability of the complementary matrix occurring in a given time slot is $(1/\eta)^{K^2}$. Thus, the number of time slots until the complementary matrix occurs is a geometric random variable with parameter $(1/\eta)^{K^2}$ and the expected delay is 
\begin{align}
\eta^{K^2} = \left( (K-1) P \phi \right)^{K^2/2} \ . 
\end{align} For time-varying magnitudes, the expected delay scales in a similar fashion with an additional penalty for waiting for the magnitudes to match. 

This delay scaling roughly corresponds to the $2^{K^2}$ independent channel realizations required by the Cadambe-Jafar beamforming scheme to attain $K/2$ degrees-of-freedom over a time-varying interference channel \cite{cj08}. The rate-delay tradeoff for linear beamforming schemes can be interpreted as the degrees-of-freedom that is attainable for a given number of independent channel realizations. This tradeoff has been characterized for $3$-user interference channels by Bresler and Tse \cite{bt09}.

In the context of ergodic alignment, the first-order question is whether the exponent of the expected delay can be significantly improved without sacrificing rate. More generally, the challenge is to design channel matching schemes that operate on the optimal tradeoff between delay and rate. See \cite{kwg10,jap10} for recent work related to these questions. 

\subsection{Practical Considerations}

As noted above, the proposed ergodic alignment scheme requires very long delays to attain half the interference-free rate. This requirement, coupled with the need for full CSIT, seems to limit the scheme to scenarios where high rates are far more valuable than low delays. However, the core idea underlying ergodic alignment, matching up complementary channel matrices, can be interpreted more broadly. For instance, one can match up complementary channels across frequencies rather than time slots. As shown by Jafar \cite{jafar12}, interference can also be completely eliminated using adjacent time slots if the direct channel gains change while the interfering channel gains remain the same. Interestingly, for this blind alignment scheme, the receivers do not need to know the channel gains, only the coherence intervals. In certain cases, such as the X channel, one can induce the desired coherence intervals simply through antenna switching \cite{wgj11}. Going beyond the wireless setting, ergodic alignment has recently been investigated as a simple network coding strategy for multiple unicast traffic \cite{rdmmjv10}. Note that in wired network coding, the ``channel'' coefficients can be freely chosen, i.e., there is no need to wait for nature to provide complementary channel gains.

\section{Recovering More Messages} \label{s:multicast}

In this section, we generalize our alignment scheme to handle the case where each receiver attempts to decode more than one message. The problem setup is largely the same as in Section \ref{s:probstateic} except that now there are $L$ transmitters, each with a single message $m_\ell$ of rate $\tilde{R}_\ell$, and $K$ receivers that want exactly $M$ messages each. For simplicity, we will assume that all messages are requested by the same number of receivers. (Note that this implicitly assumes that $\frac{KM}{L}$ is an integer.) Let $\mathcal{S}_k$ denote the set of indices of messages desired at receiver $k$ and let $\mathcal{S}_k(i)$ denote the $i^{\text{th}}$ index in the set. We now replace Definitions \ref{d:decoders} and \ref{d:rates} with the following two definitions.
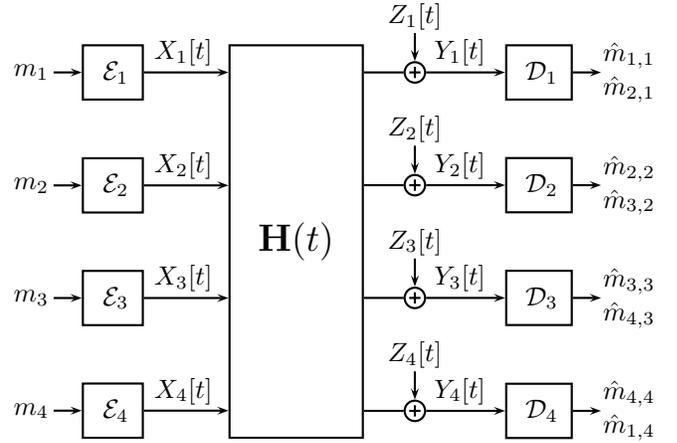
\begin{figure}[h]
\begin{center}
\psset{unit=0.75mm}
\begin{pspicture}(0,-33)(109,44)

\rput(1,32){$m_1$} \psline{->}(5,32)(10,32) \psframe(10,27)(21,37)
\rput(16,32){$\mathcal{E}_1$} \rput(28,35.5){$X_1[t]$}
\psline{->}(21,32)(36,32)

\rput(0,-20){
\rput(1,32){$m_2$} \psline{->}(5,32)(10,32) \psframe(10,27)(21,37)
\rput(16,32){$\mathcal{E}_2$} \rput(28,35.5){$X_2[t]$}
\psline{->}(21,32)(36,32)
}

\rput(0,-40){
\rput(1,32){$m_3$} \psline{->}(5,32)(10,32) \psframe(10,27)(21,37)
\rput(16,32){$\mathcal{E}_3$} \rput(28,35.5){$X_3[t]$}
\psline{->}(21,32)(36,32)
}

\rput(0,-60){
\rput(1,32){$m_4$} \psline{->}(5,32)(10,32) \psframe(10,27)(21,37)
\rput(16,32){$\mathcal{E}_4$} \rput(28,35.5){$X_4[t]$}
\psline{->}(21,32)(36,32)
}


\psframe(36,-33)(60,37)
\rput(48,2){\Large{$\mathbf{H}(t)$}}

\rput(20,0){
\psline(40,32)(47,32)
\pscircle(49,32){2} \psline{-}(48,32)(50,32)
\psline{-}(49,31)(49,33) \psline{->}(49,39)(49,34) \rput(49,42){$Z_1[t]$}
\psline{->}(51,32)(65,32) \rput(57,35.5){$Y_1[t]$}

\psline(40,12)(47,12)
\pscircle(49,12){2} \psline{-}(48,12)(50,12)
\psline{-}(49,11)(49,13) \psline{->}(49,19)(49,14) \rput(49,22){$Z_2[t]$}
\psline{->}(51,12)(65,12) \rput(57,15.5){$Y_2[t]$}

\psline(40,-8)(47,-8)
\pscircle(49,-8){2} \psline{-}(48,-8)(50,-8)
\psline{-}(49,-9)(49,-7) \psline{->}(49,-1)(49,-6) \rput(49,2){$Z_3[t]$}
\psline{->}(51,-8)(65,-8) \rput(57,-4.5){$Y_3[t]$}

\rput(0,-20){
\psline(40,-8)(47,-8)
\pscircle(49,-8){2} \psline{-}(48,-8)(50,-8)
\psline{-}(49,-9)(49,-7) \psline{->}(49,-1)(49,-6) \rput(49,2){$Z_4[t]$}
\psline{->}(51,-8)(65,-8) \rput(57,-4.5){$Y_4[t]$}
}

\psframe(65,27)(77,37) \rput(71.5,32){$\mathcal{D}_1$}
\psline{->}(77,32)(82,32)  \rput(87,35){$\hat{m}_{1,1}$} \rput(87,29){$\hat{m}_{2,1}$}

\psframe(65,7)(77,17) \rput(71.5,12){$\mathcal{D}_2$}
\psline{->}(77,12)(82,12)  \rput(87,15){$\hat{m}_{2,2}$} \rput(87,9){$\hat{m}_{3,2}$}

\psframe(65,-13)(77,-3) \rput(71.5,-8){$\mathcal{D}_3$}
\psline{->}(77,-8)(82,-8) \rput(87,-5){$\hat{m}_{3,3}$} \rput(87,-11){$\hat{m}_{4,3}$}

\rput(0,-20){
\psframe(65,-13)(77,-3) \rput(71.5,-8){$\mathcal{D}_4$}
\psline{->}(77,-8)(82,-8)\rput(87,-5){$\hat{m}_{4,4}$} \rput(87,-11){$\hat{m}_{1,4}$}
}
}

\end{pspicture}
\end{center}
\caption{Interference channel where each receiver wants $M = 2$ messages. } \label{f:multicastic}
\end{figure}

\begin{definition}[Decoders] Each receiver has a \textit{decoding function},
\begin{equation}\mathcal{D}_k: \mathbb{C}^T \rightarrow \prod_{i=1}^M \{1,2,\ldots,2^{n\tilde{R}_{\mathcal{S}_k(i)}}\} \ ,
\end{equation} that maps its length $T$ observed channel output $\{Y_k[t]\}_{t=1}^T$ into estimates $\widehat{m}_{\ell,k}$ of its desired messages $m_\ell$ for all $\ell$ such that $\ell \in \mathcal{S}_k$. 
\end{definition}

\begin{definition}[Achievable Rates] We say that a rate tuple $(R_1,R_2,\ldots,R_L)$ is \textit{achievable} if for all $\epsilon > 0$ and $T$ large enough there exist channel encoding and decoding functions $\mathcal{E}_1, \ldots, \mathcal{E}_L, \mathcal{D}_1, \ldots, \mathcal{D}_K$ such that
\begin{align}
&\tilde{R}_\ell > R_\ell - \epsilon,~~~\ell=1,2,\ldots,L \ , \\
&\P\left(\bigcup_k \bigcup_{\ell \in \mathcal{S}_k}\{\widehat{m}_{\ell,k} \neq m_\ell\}  \right) < \epsilon \ .
\end{align}
\end{definition}

 In Figure \ref{f:multicastic}, we provide a block diagram of a case with $L=4$ transmitters, $K = 4$ receivers, and message requests $\mathcal{S}_1 = \{1,2\}, \mathcal{S}_2 = \{2,3\},\mathcal{S}_3 = \{3,4\}$, and $ \mathcal{S}_4 = \{4,1\}$.

As before, all encoders retransmit their symbols at well-chosen time indices. This has the effect of giving the decoders equations with the symbols as the variables and the coefficients given by the channel. Here, it is insufficient to look for pairs of channel coefficients that exactly cancel. Since each receiver wants $M$ messages, we will need $M+1$ time slots (or dimensions). Of these, $M$ will be used for the desired messages and the remaining dimension will be used for the interfering terms. Define
\begin{equation}
\omega \triangleq \exp\bigg(\frac{j2\pi}{M+1}\bigg)
\end{equation} to be the $(M+1)^{\text{th}}$ root of unity and let $\mathbf{W}$ be the size $M+1$ discrete Fourier transform (DFT) matrix:
\begin{align}
\mathbf{W} = \left[
 \begin{array}{ccccc}
\omega^0 & \omega^0 & \omega^0  & \cdots & \omega^0  \\
\omega^0 & \omega^1 & \omega^2  & \cdots & \omega^{M} \\
\omega^0 & \omega^2  & \omega^4  & \cdots & \omega^{2M} \\
 \vdots & \vdots & \vdots & \ddots & \vdots \\
\omega^0 & \omega^{M} & \omega^{2M}  & \cdots & \omega^{M^2}   
\end{array}
\right] .
\end{align}
The inverse DFT matrix has the following form: 
\begin{align*}
\mathbf{W}^{-1} = \frac{1}{M+1}\left[
 \begin{array}{ccccc}
\omega^0 & \omega^0 & \omega^0 &  \cdots & \omega^0  \\
\omega^0 & \omega^{-1} & \omega^{-2}  &    \cdots & \omega^{-M} \\
\omega^0 & \omega^{-2}  & \omega^{-4}  &  \cdots & \omega^{-2M} \\
 \vdots & \vdots & \vdots & \ddots & \vdots \\
\omega^0 & \omega^{-M} &  \omega^{-2M} &\cdots & \omega^{-M^2}   
\end{array}
\right] 
\end{align*}

First, consider the following idealized scenario. As in the introduction, assume that each transmitter sends signals $X_1, X_2,\ldots, X_L$ at time $t_1$ under channel matrix $\Hhb = \{h_{k\ell}\}$. The transmitters then wait for channel coefficients satisfying
\begin{align}
h_{k\ell}[t_n] = \begin{cases} 
\omega^{(i-1)(n-1)} h_{k\ell}[t_1]& \ell = \mathcal{S}_k(i), \\
\omega^{M(n-1)}h_{k\ell}[t_1] & \ell \notin \mathcal{S}_k . 
\end{cases}
\end{align} for $n = 2,\ldots, M+1$ and resend $X_1, X_2,\ldots, X_L$ during these time slots. Assume that receiver $k$ wants the first $M$ messages\footnote{For any other choice of $\mathcal{S}_{k}$ simply replace the transmitter indices $1,\ldots, M$ with $\mathcal{S}_k(1), \ldots, \mathcal{S}_k(M)$.}. Then, the channel observations at receiver $k$ can be written in vector form as
\begin{align}
\left[\begin{array}{c}
Y_k[t_1] \\
Y_k[t_2] \\
\vdots \\
Y_k[t_{M+1}]
\end{array}\right] = \mathbf{W} \left[\begin{array}{c}
h_{k1}[t_1] X_{1} \\
\vdots \\
h_{k M}[t_1] X_{M} \\
{\displaystyle \sum_{\ell \notin \mathcal{S}_k} h_{k\ell}[t_1]X_\ell}
\end{array}\right] + \left[\begin{array}{c}
Z_k[t_1] \\
Z_k[t_2] \\
\vdots \\
Z_k[t_{M+1}]
\end{array}\right] \nonumber \ . 
\end{align} That is, each desired signal is assigned to a unique DFT vector. All of the undesired signals are assigned to a single DFT vector that is orthogonal from the others. As a result, the receiver can apply the inverse DFT matrix to its vector of observations to extract its desired signals,
\begin{align}
 \mathbf{W}^{-1}\left[\begin{array}{c}
Y_k[t_1] \\
Y_k[t_2] \\
\vdots \\
Y_k[t_{M+1}]
\end{array}\right] =  \left[\begin{array}{c}
h_{k1}[t_1] X_{1} \\
\vdots \\
h_{k M}[t_1] X_{M} \\
{\displaystyle \sum_{\ell \notin \mathcal{S}_k} h_{k\ell}[t_1]X_\ell}
\end{array}\right] + \left[\begin{array}{c}
\tilde{Z}_k[t_1] \\
\tilde{Z}_k[t_2] \\
\vdots \\
\tilde{Z}_k[t_{M+1}]
\end{array}\right] \nonumber 
\end{align} where the $\tilde{Z}_k[t_n]$ are transformed noise terms,
\begin{align}
\tilde{Z}_k[t_n] =\frac{1}{M+1} \sum_{m=1}^{M+1} \omega^{-(i-1)(n-1)} Z_k[t_m] \ , \label{e:transformednoise}
\end{align} that are i.i.d. circularly symmetric Gaussian random variables with mean zero and variance $1/(M+1)$. Of course, we cannot afford to wait until the channel coefficients match precisely; instead, we will match up time slots based on quantized channel coefficients. The next theorem formalizes the scheme described above.

\begin{theorem}\label{t:multicast}
For the time-varying Gaussian interference channel (as defined in Section \ref{s:probstateic}) where receiver $k$ wants $M$ messages $\{m_\ell : \ell \in \mathcal{S}_k\}$, the rates
\begin{align}
R_{\ell} = \min_{k: \ell \in \mathcal{S}_k}\frac{1}{M+1}\E\big[\log{(1 + (M+1) |h_{k\ell}|^2 P)}\big] \ .
\end{align} are achievable for $\ell = 1,2,\ldots,L$. \end{theorem}

\begin{IEEEproof}
Choose $\epsilon > 0$. We will divide up the $T$ channel uses into $M+1$ consecutive intervals of equal length. We quantize all channel realizations using the scheme described in Section \ref{s:chanquant}. Applying Lemma \ref{l:typical} with $N = M+1$, it follows that all $M+1$ blocks of $T/(M+1)$ channel uses are $\gamma$-typical with probability at least $(1 - \frac{\epsilon}{2})$ (with $\gamma$ to be specified later). By Definition \ref{d:typical}, this means that the number of occurrences of each possible quantized channel matrix in each interval is bounded as follows:
\begin{align}
\frac{T}{M+1}\Big(p_{\Hhb}\big(\Hhs\big) - \gamma \Big) \leq \#\big(\Hhs | \Hhb^{(n)}\big) \leq \frac{T}{M+1}\Big(p_{\Hhb}\big(\Hhs\big) + \gamma \Big)  \nonumber
\end{align} for all $\Hhs \in \mathcal{H}$. 

If any block is not $\gamma$-typical, we declare an error. Otherwise, we know that each quantized matrix will occur at least $\frac{T}{M+1}\left(p_{\Hhb}\big(\Hhs\big) - \gamma \right)$ times in each interval. A time slot $t$ in an interval is useable unless:
\begin{enumerate}
\item The channel matrix $\mathbf{H}[t]$ contains one or more elements with magnitude larger than $h_{\text{MAX}}$. 
\item The channel matrix $\mathbf{H}[t]$ does not violate the threshold but has already occurred at least  $\frac{T}{M+1}\left(p_{\Hhb}\big(\Hhs\big)  - \gamma \right)$ times. 
\end{enumerate}
Assuming the intervals are $\gamma$-typical, the number of useable time slots per interval is \begin{align}
 &\left\lfloor \frac{T}{M+1}\sum_{\Hhs: \hat{h}_{k\ell} \neq \Gamma }\Big(p_{\Hhb}\big(\Hhs\big) - \gamma \Big) \right\rfloor \nonumber \\ &= \left\lfloor \frac{T}{M+1}\left(1 - \rho - \gamma(\kappa \eta)^{KL}\right)\right\rfloor.
\end{align}

Each encoder employs an independent codebook $\mathcal{C}_\ell$ with rate $\tilde{R}_\ell$ and length chosen to match the number of useable time slots per block. The codewords are generated i.i.d. from a circularly symmetric Gaussian distribution with variance slightly less than $P$. 

During the first interval, each transmitter sends out a new symbol from its codeword during each useable time slot $t_1$ and records the corresponding quantized channel matrices $\Hhb[t_1] = \{\hat{h}_{k\ell}[t_1]\}$. We match up each useable time slot $t_1$ from the first interval with useable time slot $t_n$ from the $n^{\text{th}}$ interval for $n = 2,\ldots, M+1$ such that
\begin{align}
\hat{h}_{k\ell}[t_n] = \begin{cases} 
\omega^{(i-1)(n-1)} \hat{h}_{k\ell}[t_1]& \ell = \mathcal{S}_k(i), \\
\omega^{M(n-1)} \hat{h}_{k\ell}[t_1] & \ell \notin \mathcal{S}_k 
\end{cases}
\end{align} where $\mathcal{S}_k(i)$ is the $i^{\text{th}}$ message requested by receiver $k$. Note that this matching can be performed using only causal channel knowledge by greedily matching up time slots from intervals $n=2,\ldots,M+1$ with time slots from the first interval in the order in which they occur. To ensure that $\omega^{(i-1)(n-1)}\hat{h}_{k\ell}$ corresponds to a valid quantization cell, we constrain the number of angles (given by $\eta$) to be a multiple of $M+1$. Owing to the symmetry of the uniform phase assumption, all useable time slots will be successfully matched.

Note that each receiver essentially observes a DFT of its desired signals and the interference. Thus, by applying the inverse DFT, each receiver can see its desired signals through nearly interference-free channels. Specifically, the quantized channel coefficients satisfy
\begin{align}
\frac{1}{M+1} \sum_{n=1}^{M+1} \omega^{-(i-1)(n-1)} \hat{h}_{k\ell}[t_n] = \begin{cases} \hat{h}_{k\ell}[t_1]& \ell = \mathcal{S}_k(i), \\
0 & \ell \neq \mathcal{S}_k(i) \ . \nonumber
\end{cases}
\end{align} Therefore, applying the same transformation to channel observations from matched time slots $t_1, \ldots, t_{M+1}$,
\begin{align}
\tilde{Y}_{km}[t_1] = \frac{1}{M+1} \sum_{n=1}^{M+1} \omega^{-(i-1)(n-1)} Y_{k}[t_n] \ , 
\end{align} will yield nearly interference-free channels from the transmitters in $\mathcal{S}_k$ to receiver $k$.

Using Lemma \ref{l:quantbounds}, we have that
\begin{align}
&\bigg|\frac{1}{M+1} \sum_{n=1}^{M+1} \omega^{-(i-1)(n-1)} {h}_{k\mathcal{S}_k(i)}[t_n]  \bigg| \\
&\geq \max\Big(\big|\hat{h}_{k\mathcal{S}_k(i)}[t_1]\big| -  \delta,0\Big) \\
&\geq \max\Big(\big|{h}_{k\mathcal{S}_k(i)}[t_1]\big| -  2\delta,0\Big) 
\end{align} where $\delta$ is the maximum distance between any two points in a quantization cell. It follows that the signal power in $\tilde{Y}_{ki}[t_1] $ is at least
\begin{align}
\Big(\big|h_{k\mathcal{S}_k(i)}[t_1]\big| - 2 \delta\Big)^2 P
\end{align} if $\big|h_{k\mathcal{S}_k(i)}[t_1]\big| > 2\delta$. Applying Lemma \ref{l:quantbounds}, we get that the interference power from each transmitter $\ell \neq \mathcal{S}_k(i)$ is at most $\delta^2 P$. The noise power is exactly $1/(M+1)$ as shown in (\ref{e:transformednoise}). Thus, the resulting channel from each transmitter $\ell$ to receiver $k$ for $\ell = \mathcal{S}_k(i)$ has signal-to-interference-and-noise ratio no less than \begin{align}
\mathsf{SINR}_{k\ell} \geq  \frac{\Big(\big|h_{k\ell}[t_1]\big| - 2 \delta\Big)^2 P}{K \delta^2 P + (M+1)^{-1}   } \ . \nonumber
\end{align} if $\big|h_{k\ell}[t_1]\big| > 2\delta$. Choosing $\delta$ small enough (by making $\kappa$ and $\eta$, the quantization parameters, large enough), we make the signal-to-interference-and-noise ratios satisfy 
\begin{align}
\mathsf{SINR}_{k\ell} \geq (M+1) | h_{k\ell} [t_1] |^2 P - \lambda
\end{align} for some $\lambda > 0$ to be specified later.

By choosing $h_{\text{MAX}}$ large enough and $\gamma$ small enough, we can ensure there are at least $\frac{T(1-\lambda)}{M+1}$ useable time slots. For $T$ large enough, receiver $k$ can decode the message $m_\ell$ from transmitter $\ell \in \mathcal{S}_k$ with probability of error $\frac{\epsilon}{2KM}$ if
\begin{align*}
\tilde{R}_\ell &\leq \frac{1-\lambda}{M+1} \E\big[ \log{(1 + (M+1) | h_{k\ell} [t_1] |^2 P - \lambda)}\big] - \frac{\epsilon}{3} \ .
\end{align*} Choosing $\lambda$ small enough, we get 
\begin{align}
\tilde{R}_\ell &\leq \frac{1}{M+1} \E\big[ \log{(1 + (M+1) | h_{k\ell} [t_1] |^2 P)}\big] - \frac{2\epsilon}{3} \ .
\end{align} Thus, by the union bound, all receivers can decode their messages with probability of error $\frac{\epsilon}{2}$ if
\begin{align*}
\tilde{R}_\ell &\leq \frac{1}{M+1} \min_{k: \ell \in \mathcal{S}_k} \E\big[ \log{(1 + (M+1) | h_{k\ell} [t_1] |^2 P)}\big] - \frac{2\epsilon}{3} \ .
\end{align*} Recall also that with probability $\frac{\epsilon}{2}$ the channel is not $\gamma$-typical so the total probability of error is less than $\epsilon$. Therefore, there must exist a set of good fixed codebooks which we can expurgate to meet the power constraint with an additional rate loss of at most $\epsilon/3$.
\end{IEEEproof}

As before, we have not optimized the power allocation using the transmitters' knowledge of the channel realizations. 

From the upper bound (\ref{e:multicastupper}) in Appendix \ref{s:outer}, it follows that (for symmetric rates), it is impossible to attain a pre-log factor greater than $1/(M+1)$.

\begin{remark}
If we simply extended the scheme from Theorem \ref{t:icalign} and cancelled out the interference from each desired signal one-by-one, we would not attain the same power gain. Specifically, assume that at time $t_{n}$ we flip the channel coefficients from transmitter $\ell = \mathcal{S}_k(n-1)$ to receiver $k$,
\begin{align}
\hat{h}_{k\ell}[t_n] = \begin{cases} 
\hat{h}_{k\ell}[t_1]& \ell = \mathcal{S}_k(n-1), \\
- \hat{h}_{k\ell}[t_1] & \ell \neq \mathcal{S}_k(n-1) \ . 
\end{cases}
\end{align} The receivers can then simply add together times $t_1$ and $t_{n}$ to get a clean channel from transmitter $\ell = \mathcal{S}_k(n-1)$ to obtain their $(n-1)^{\text{th}}$ desired message. However, this will only yield a power gain of $2$ instead of the full gain of $M+1$,
\begin{align}
R_{\ell} = \min_{k: \ell \in \mathcal{S}_k}\frac{1}{M+1}\E\Big[\log{\big(1 + 2 |h_{k\ell}|^2 P \big)}\Big].
\end{align}
\end{remark} 

In Figure \ref{f:multicastplot}, we have plotted the performance of the scheme from Theorem \ref{t:multicast} over the network in Figure \ref{f:multicastic} with i.i.d. Rayleigh fading. The upper bound is from (\ref{e:multicastupper}) in Appendix \ref{s:outer} and the time division scheme is from (\ref{e:tdma}) for $4$ users. The ergodic alignment scheme has the same $1/3$ slope as the upper bound whereas the time division scheme has a slope of $1/4$. The gap between alignment and time division becomes more pronounced if we increase the ratio between transmitters $L$ and the number of desired messages $M$.

\begin{figure}[h]
\centering
\includegraphics[width=3.8in]{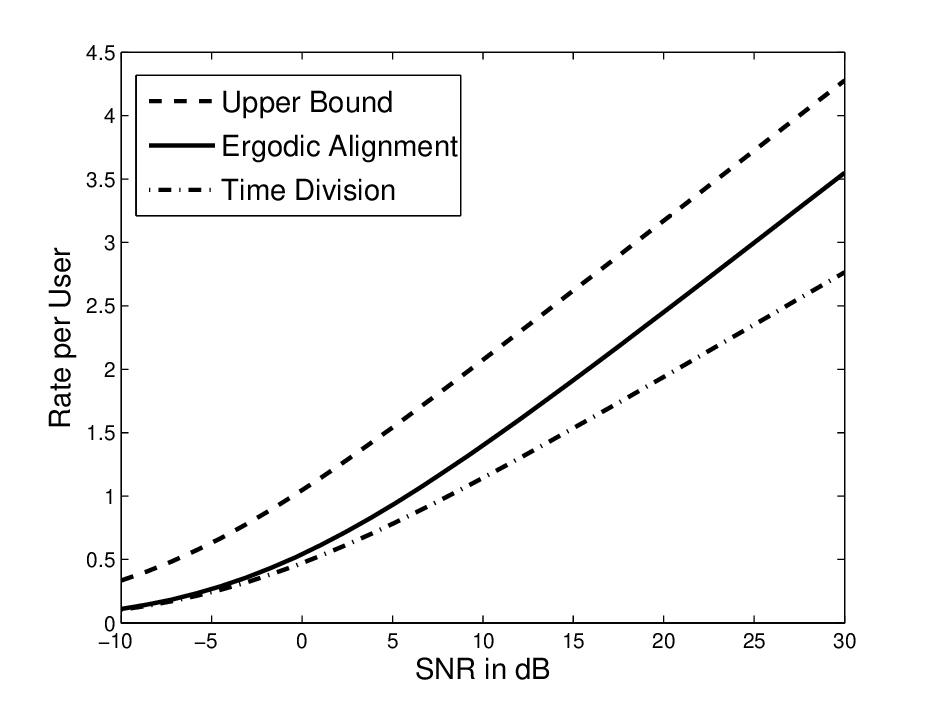}
\caption{Rate per user for the network in Figure \ref{f:multicastic} with i.i.d. Rayleigh fading.}\label{f:multicastplot}
\end{figure}

\begin{remark}
Very recent work by Ke \textit{et al.} has determined the degrees-of-freedom region for an interference channel where each receiver requests an arbitrary subset of the transmitted messages \cite{krwy12}. 
\end{remark}

\section{X Message Set} \label{s:xchannel}

We now turn to a variant of the interference channel, the X channel, that has garnered significant attention \cite{js08,mmk08,cj09ITb}. In this scenario, there are $L$ transmitters and $K$ receivers and each transmitter has an independent message for each receiver. For the single antenna case, Cadambe and Jafar showed that the sum degrees-of-freedom is $\frac{LK}{L+K-1}$ using interference alignment \cite{cj09ITb}. Here, we extend this result to the finite SNR regime for the special case of $K = 2$ receivers. Let $m_{\ell 1}$ and $m_{\ell 2}$ denote the messages sent from the $\ell^{\text{th}}$ transmitter to the first and second receiver, respectively. Each message has rate $\tilde{R}_{\ell k}$. Figure \ref{f:xchan2} is a block diagram of an X message set for $K = 2$ transmitters and $L = 2$ receivers.
\begin{figure}[h]
\begin{center}
\psset{unit=0.75mm}
\begin{pspicture}(-8,7)(115,42)

\rput(-2,32){$\begin{array}{c}m_{11}\\m_{12}\end{array}$} \psline{->}(3,32)(8,32) \psframe(8,27)(18,37)
\rput(13.5,32){$\mathcal{E}_1$} \rput(26,35.5){$X_1[t]$}
\psline{->}(18,32)(34,32)

\rput(0,-20){
\rput(-2,32){$\begin{array}{c}m_{21}\\m_{22}\end{array}$} \psline{->}(3,32)(8,32) \psframe(8,27)(18,37)
\rput(13.5,32){$\mathcal{E}_2$} \rput(26,35.5){$X_2[t]$}
\psline{->}(18,32)(34,32)
}

\psframe(34,7)(60,37)
\rput(47,22){\Large{$\mathbf{H}(t)$}}

\rput(18,0){
\psline(42,32)(47,32)
\pscircle(49,32){2} \psline{-}(48,32)(50,32)
\psline{-}(49,31)(49,33) \psline{->}(49,39)(49,34) \rput(49,42){$Z_1[t]$}
\psline{->}(51,32)(67,32) \rput(59,35.5){$Y_1[t]$}

\psline(42,12)(47,12)
\pscircle(49,12){2} \psline{-}(48,12)(50,12)
\psline{-}(49,11)(49,13) \psline{->}(49,19)(49,14) \rput(49,22){$Z_2[t]$}
\psline{->}(51,12)(67,12) \rput(59,15.5){$Y_2[t]$}
\psframe(67,27)(77,37) \rput(72.5,32){$\mathcal{D}_1$}
\psline{->}(77,32)(82,32)  \rput(87,32){$\begin{array}{c}\hat{m}_{11}\\\hat{m}_{21}\end{array}$}

\psframe(67,7)(77,17) \rput(72.5,12){$\mathcal{D}_2$}
\psline{->}(77,12)(82,12) \rput(87,12){$\begin{array}{c}\hat{m}_{12}\\\hat{m}_{22}\end{array}$}
}

\end{pspicture}
\end{center}
\caption{X message set for $K = 2$ transmitters and $L = 2$ receivers. } \label{f:xchan2}
\end{figure}
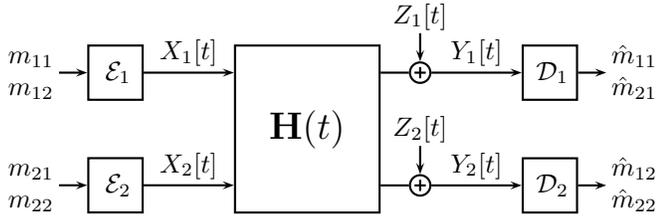

Unlike in our previous schemes, we cannot hope for the channel to generate an independent coefficient for every message. Transmitters should instead separate their messages by premultiplying them by phases. This leaves us with fewer variables to work with to align the interference at every receiver. For simplicity, assume that each transmitter splits its power equally between its messages $m_{\ell 1}$ and $m_{\ell 2}$. Each of these messages is mapped to a codeword whose symbols are represented by $X_{\ell 1}$ and $X_{\ell 2}$, respectively. 

We first introduce our scheme in an idealized setting where the transmitters wait for channel coefficients that precisely match. At time $t_1$, each transmitter sends $X_\ell = X_{\ell 1} + X_{\ell 2}$ and records the resulting channel realization $\mathbf{H}[t_1] = \{ h_{k\ell}[t_1] \}$. The transmitters then wait for time slots $t_2,\ldots, t_{L+1}$ satisfying
\begin{align}
h_{1\ell}[t_n] &= h_{1 \ell}[t_1] \\
h_{2\ell}[t_n] &= \omega^{-(\ell - 1)(n-1)}h_{2 \ell}[t_1] 
\end{align} where $\omega = \exp(j 2\pi/(L+1))$. During these time slots, the transmitters send
\begin{align}
X_{\ell}[t_n] &= \omega^{(\ell-1)(n-1)} X_{\ell 1} + \omega^{L(n-1)} X_{\ell 2} \ . 
\end{align} 

The resulting channel outputs at receiver $1$ can be written in vector form as
\begin{align}
\left[\begin{array}{c}
Y_1[t_1] \\
Y_1[t_2] \\
\vdots \\
Y_1[t_{L+1}]
\end{array}\right] = \mathbf{W} \left[\begin{array}{c}
h_{11}[t_1] X_{11} \\
\vdots \\
h_{1 L }[t_1] X_{L1} \\
{\displaystyle \sum_{\ell=1}^L h_{1 \ell}[t_1]X_{\ell 2}}
\end{array}\right] + \left[\begin{array}{c}
Z_1[t_1] \\
Z_1[t_2] \\
\vdots \\
Z_1[t_{L+1}]
\end{array}\right] \nonumber \ . 
\end{align} Similarly, the channel outputs at receiver $2$ are
\begin{align}
\left[\begin{array}{c}
Y_2[t_1] \\
Y_2[t_2] \\
\vdots \\
Y_2[t_{L+1}]
\end{array}\right] = \mathbf{W} \left[\begin{array}{c}
{\displaystyle \sum_{\ell=1}^L h_{2 \ell}[t_1]X_{\ell 1}} \\
h_{2L}[t_1] X_{L2} \\
\vdots \\
h_{2 1 }[t_1] X_{12} 
\end{array}\right] + \left[\begin{array}{c}
Z_2[t_1] \\
Z_2[t_2] \\
\vdots \\
Z_2[t_{L+1}]
\end{array}\right] \nonumber \ . 
\end{align} As in Section \ref{s:multicast}, each desired signal is assigned to a unique DFT vector and all the interfering terms are grouped into the remaining vector. Following the steps of the proof of Theorem \ref{t:multicast}, we can arrive at the following theorem.

\begin{theorem} \label{t:xchannel}
For the X message set with $M = 2$ receivers, the following rates are achievable over the time-varying Gaussian interference channel defined in Section \ref{s:probstateic},
\begin{align}
R_{\ell m } = \frac{1}{L+1}\E\bigg[\log\bigg(1 +  \frac{(L+1) |h_{k \ell}|^2P}{2}\bigg)\bigg] \ . 
\end{align} 
\end{theorem}

\begin{remark} Unfortunately, the scheme above does not directly generalize to $L > 2$ receivers. The key issue is that each symbol travels through an effective channel to each receiver, with phases determined by our channel matching scheme. In the interference channel, these phases can be set to arbitrarily values. For the X channel, there are $LK$ symbols that are each seen by $K$ receivers. If we demand specific phases for each effective channel from symbol to receiver, we will end up with $LK^2$ constraints. Each transmitter can pre-multiply the symbols by phases, leading to $LK$ free variables, and we can wait for phases on the $LK$ channel gains. Overall, we have $LK^2$ constraints and $2LK$ free variables, meaning that the problem becomes overconstrained when $L > 2$. 
\end{remark}

\section{Time-Varying Finite Field Interference Channel} \label{s:finitefield}
 
For the Gaussian case, it is sufficient to match up channel matrices and add up the resulting channel outputs. The simplicity of this strategy is in some ways an artifact of the Gaussian setting. In general, the receivers may need to perform a decoding step prior to combining the observed signals to avoid noise build-up. In this section, we consider a finite field interference channel with fast fading and derive the entire capacity region. Each receiver groups together time instances with the same channel coefficients and decodes a function of the messages, using a linear code. By combining two appropriately chosen functions, the interference can be completely removed.

The problem statement is identical to that in Section \ref{s:probstateic} except for the channel model. We assume that all operations are carried out over a finite field $\mathbb{F}_q$. Let $\oplus$ and $\bigoplus$ denote addition and summation over $\mathbb{F}_q$, respectively.

\begin{definition}[Channel Model]
We assume that the channel inputs and outputs take values on the same finite field $\mathbb{F}_q$. The channel output observed by each receiver is a noisy linear combination of its inputs:
\begin{align}
Y_k[t] = \bigoplus_{\ell=1}^K{h_{k \ell}[t] X_\ell[t]} \oplus Z_k[t]
\end{align} where the $h_{k\ell}[t]$ are time-varying channel coefficients and $Z_k[t]$ is additive i.i.d. noise drawn from a distribution that takes values uniformly on $\{1,2,\ldots,q-1 \}$ with probability $\nu$ and is zero otherwise. The entropy of this distribution is 
$$H(Z) = - \nu \log{\nu} - (1- \nu) \log{(1 - \nu)} + \nu \log{(q-1)} \ .$$

\begin{remark} The assumed symmetry of the noise distribution across its non-zero values plays an important role in our capacity proof. That is, our outer bound relies on the fact that scaling the noise by a non-zero number does not alter its distribution.
\end{remark}

 We assume that at each time step each channel coefficient is drawn independently and uniformly from $\mathbb{F}_q \setminus \{0\}$. The transmitters and receivers are given access to the channel realizations causally. That is, before time $t$, each transmitter and receiver is given $h_{k\ell}[t]$ for all $k$ and $\ell$. Let $\mathbf{H}[t] = \{h_{k\ell}[t]\}$ denote the matrix of channel coefficients.
\end{definition}

\begin{remark} Using counting arguments, we can extend our results to the case where the channel coefficients are allowed to equal zero with some probability. However, this considerably complicates the description of the capacity region. 
\end{remark}

The basic idea underlying our scheme is to add together two well-chosen channel outputs such that the interference exactly cancels out. As before, we can will match a channel matrix $\mathbf{H}$ with a complementary matrix
\begin{align}
g(\mathbf{H}) &\triangleq 
\left[
\begin{array}{cccc}
 1\oplus(-h_{11}) & -h_{12}  & \cdots& -h_{1K}  \\
 -h_{21} & 1\oplus(-h_{22})  & \cdots& -h_{2K}  \\
 \vdots & \vdots & \ddots & \vdots \\
 -h_{K1} & -h_{K2} & \cdots & 1\oplus(-h_{KK})
\end{array}
\right]  \nonumber
\end{align} so that $\mathbf{H} \oplus g(\mathbf{H}) = \mathbf{I}$. However, for the finite field model, if we directly sum up the observations from a given channel matrix and its complement, we will accumulate noise. As it turns out, it is better to group together time slots based on their channel realization and send a \textit{linear function} of the messages to each receiver using a linear code. This technique, sometimes referred to as computation coding \cite{ng07IT}, is reviewed in detail in Appendix \ref{s:compute}. We then match up linear functions so that the receivers can solve for their desired messages.

Since the channel coefficients are drawn from a discrete alphabet, we can define typicality without resorting to quantization. Assume that the $T$ channel uses are split into two consecutive blocks of equal length. Let 
\begin{align}
\#\Big(\mathsf{H}|\mathbf{H}^{(n)}\Big) \triangleq \bigg|\bigg\{ t : \mathbf{{H}}[t] = \mathsf{{H}},~1+ \frac{(n-1)T}{2}  \leq t \leq \frac{nT}{2} \bigg\}\bigg| \nonumber
\end{align} be the number of channel matrices within the $n^{\text{th}}$ block that are equal to $\mathsf{{H}} \in \mathbb{F}^{K \times K}_q$. The definition of $\gamma$-typicality is the same as that given in Definition \ref{d:typical}. From Lemma \ref{l:typical}, it follows that, for any $\epsilon > 0$ and $T$ large enough, both blocks are $\gamma$-typical with probability at least $1-\epsilon$.

If a transmitter-receiver pair had the channel to itself, it can achieve an interference-free rate of $\log{q} - H(Z)$. We will now show that all users can achieve half the interference-free rate simultaneously.
\begin{theorem}\label{t:fieldsymmetric}
For the time-varying finite field interference channel, the following rates are achievable 
\begin{align}
R_k = \frac{1}{2}(\log{q} - H(Z)) \ .
\end{align} 
\end{theorem}
\begin{proof}
For any $\epsilon > 0$, let $\gamma$ be a small positive constant that will be chosen later to satisfy our rate requirement. Using Lemma \ref{l:typical}, choose $T$ large enough such that $\mathbf{H}^{(1)}$ and $\mathbf{H}^{(2)}$ are both $\gamma$-typical with probability $1-\frac{\epsilon}{2}$. Let $\mathcal{F} = \{ \mathbb{F}_q \setminus \{0\}\}^{K \times K}$ denote the channel matrix alphabet. Now, condition on the event that both blocks are $\gamma$-typical. Since the channel coefficients are i.i.d. and uniform, the probability of any channel $\mathsf{H} \in \mathcal{F}$ is $| \mathcal{F} |^{-1}$. Since $\mathbf{H}^{(n)}$ is $\gamma$-typical we have that for every $\mathsf{H} \in \mathcal{F}$:
\begin{align}
\frac{T}{2} \left(\frac{1}{|\mathcal{F}|}-\gamma\right)  \leq \#(\mathsf{H}|\mathbf{H}^{(n)}) \leq   \frac{T}{2}\left(\frac{1}{|\mathcal{F}|}+\gamma\right)  \ . \end{align} 

Let $\mathcal{T}^{(n)}_{\mathsf{H}}$ denote the first $ \frac{T}{2}(| \mathcal{F} |^{-1}-\gamma )$ time indices from the $n^{\text{th}}$ block with channel realization $\mathsf{H} \in \mathcal{F}$. We will ignore all other time slots which reduces the rate by at most a factor $(1 -\gamma)$. Each transmitter splits its message into many distinct chunks, one for each channel realization $\mathsf{H}$. Let $\mathbf{w}_{\ell\mathsf{H}} \in \mathbb{F}_q^\kappa$ be the chunk intended for realization $\mathsf{H}$. Assuming the chunks are all $\gamma$-typical, the length of each chunk is  
\begin{align}
\kappa =  \frac{T}{2}\bigg(\frac{1}{|\mathcal{F}|}-\gamma \bigg)\frac{\log q - H(Z) - \epsilon/2}{\log q} \ .
\end{align}

Using the computation code described in Appendix \ref{s:compute}, each transmitter $\ell$ sends its message $\mathbf{w}_{\ell\mathsf{H}}$ during the time indices in $\mathcal{T}^{(1)}_{\mathsf{H}}$. Receiver $k$ makes an estimate $\mathbf{\hat{u}}_{k\mathsf{H}}$ of 
$$\mathbf{u}_{k\mathsf{H}} =\bigoplus_{\ell=1}^K{h_{k\ell} \mathbf{w}_{\ell\mathsf{H}}} \ .$$ Each transmitter then employs a computation code with the \textit{same messages} $\mathbf{w}_{\ell \mathsf{H}}$ over the time indices $\mathcal{T}^{(2)}_{g(\mathsf{H})}$ in the second block corresponding to the complementary matrix $g(\mathsf{H})$. Receiver $k$ then makes an estimate $\mathbf{\hat{v}}_{k\mathsf{H}}$ of
$$\mathbf{v}_{k\mathsf{H}} =(1 \oplus (-h_{kk}))\mathbf{w}_{k\mathsf{H}} \oplus \left(- \bigoplus_{\ell \neq k}{h_{k\ell} \mathbf{w}_{\ell \mathsf{H}}}\right)\ .$$ By Lemma \ref{l:compute}, for $T$ large enough, the total probability of error for all computation codes is upper bounded by $\epsilon/2$. 

After collecting these (estimates of) linear functions, receiver $k$ makes an estimate of $\mathbf{w}_{k\mathsf{H}} $ by simply adding up the two equations to get $$\mathbf{\hat{w}}_{k\mathsf{H}} = \mathbf{\hat{u}}_{k\mathsf{H}} \oplus \mathbf{\hat{v}}_{k\mathsf{H}} \ . $$ 

The total number of bits encoded into the chunks across all $| \mathcal{F}|$ channel realizations is
\begin{align}
\frac{T}{2}\big(1 - | \mathcal{F}| \gamma \big)(\log q - H(Z) - \epsilon/2) \ . 
\end{align} Normalizing by $T$ and taking $\gamma$ small enough, the rate per transmitter is $\frac{1}{2}(\log{q} - H(Z)) - \epsilon$. The probability that either block is atypical is less than $\epsilon/2$ and the probability of error over the computation code is less than $\epsilon/2$ for $T$ large enough so the total probability of error is less than $\epsilon$ as desired. 
\end{proof}

We now use the scheme from Theorem \ref{t:fieldsymmetric} to establish the following achievable rate region.

\begin{theorem}\label{t:fieldalignment}
For the time-varying finite field interference channel, any rate tuple $(R_1,\ldots, R_K)$, satisfying the following inequalities is achievable:
\begin{align}
R_\ell + R_k \leq \log{q} - H(Z),~~~~\forall k \neq \ell. \label{e:fieldachievable}
\end{align}
\end{theorem} First, we will give an equivalent description of this rate region and then show that any rate tuple can be achieved by time sharing the symmetric rate point from Theorem \ref{t:fieldsymmetric} and a single user transmission scheme.

\begin{lemma}\label{l:fieldequivalent}
Assume, without loss of generality, that the users are labeled according to rate in descending order, so that \mbox{$R_1 \geq R_2 \geq \cdots \geq R_K$}. The achievable rate region from Theorem \ref{t:fieldalignment} is equivalent to the following rate region:
\begin{align}
R_1 &\leq \log{q} - H(Z) \\
R_k & \leq \min\Big(\log{q} - H(Z) - R_1, \frac{1}{2}\big(\log{q} - H(Z)\big) \Big),~k \geq 2\nonumber\end{align}
\end{lemma}
\begin{proof}
The key idea is that only one user can achieve a rate higher than $\frac{1}{2}(\log{q} - H(Z) )$. From (\ref{e:fieldachievable}), we must have that $R_1 + R_k \leq \log{q} - H(Z) $ so if $R_1 > \frac{1}{2}(\log{q} - H(Z) )$ all other users must satisfy $R_k \leq \log{q} - H(Z) - R_1$. If $R_1 \leq \frac{1}{2}(\log{q} - H(Z) )$, then we have that $R_k \leq \frac{1}{2}(\log{q} - H(Z) )$ for all other users since the rates are in descending order.
\end{proof}
\begin{proof}[Proof of Theorem \ref{t:fieldalignment}]
We show that the equivalent rate region developed by Lemma \ref{l:fieldequivalent} is achievable by time-sharing. First, we consider the case where $R_1 > \frac{1}{2}(\log{q} - H(Z))$. Let 
\begin{align}
\beta = 2\left(1 - \frac{R_1}{\log{q} - H(Z)}\right) \ . 
\end{align} We allocate $\beta T$ channel uses to the symmetric scheme from Theorem \ref{t:fieldsymmetric}. For, the remaining $(1- \beta)T$ channel uses, users $2$ through $K$ are silent, and user $1$ employs a capacity-achieving point-to-point channel code. This results in user $1$ achieving its target rate $R_1$:
\begin{align}
&\frac{\beta(\log{q} - H(Z))}{2} + (1-\beta)(\log{q} - H(Z))\\
&=\log{q} - H(Z) - R_1 - \log{q} + H(Z) + 2 R_1 = R_1
\end{align} and users $2$ through $K$ achieving $R_k = \log{q} - H(Z) - R_1$. If $R_1 \leq \frac{1}{2}(\log{q} - H(Z))$, we can achieve any rate point with the use of the symmetric scheme from Theorem \ref{t:fieldsymmetric}. 
\end{proof}

Finally, we will give an upper bound using the techniques in \cite{cj08} to show that the achievable rate region in Theorem \ref{t:fieldalignment} is the capacity region.

\begin{theorem}
For the time-varying finite field interference channel, the capacity region is the set of all rate tuple $(R_1, \ldots, R_K)$ satisfying
\begin{align}
R_\ell + R_k \leq \log{q} - H(Z),~~~~\forall k \neq \ell.
\end{align}
\end{theorem} 
\begin{proof}
The required upper bound follows from steps similar to those in Appendix II of \cite{cj08}. Without loss of generality, we upper bound the rates of users $1$ and $2$. Note that the capacity of the interference channel only depends on the noise marginals. Thus, we can assume that $Z_1[t] = h_{12}[t](h_{22}[t])^{-1} Z_2[t]$ due to the symmetry of the noise distribution. Multiplying $Y_2[t]$ by a non-zero factor does not change the capacity so let $\tilde{Y}_2[t] = h_{12}[t](h_{22}[t])^{-1} Y_{2}[t]$. 

We give the receivers full access to the messages from users $3$ through $K$ as this can only increase their respective rates. Assume that the corresponding signals $X_3[t],\ldots,X_K[t]$ have been eliminated from $Y_1[t]$ and $\tilde{Y}_2[t]$ below. We also give receiver $2$ access to $m_1$. Let $\epsilon_T = 1 + (R_1 + R_2)p_{\text{error}} $ where $p_{\text{error}}$ is the probability of error.  From Fano's inequality, we have that $T(R_1 + R_2)$ is upper bounded as follows:
\begin{align}
&T(R_1 + R_2) \nonumber \\
&\leq  I\Big(m_2;m_1,\{\tilde{Y}_2[t]\}_{t=1}^T\Big)+ I\Big(m_1;\{Y_1[t]\}_{t=1}^T\Big)  +T\epsilon_T \nonumber \\
&=  I\Big(m_2;\{\tilde{Y}_2[t]\}_{t=1}^T \Big| m_1 \Big)+ I\Big(m_1;\{Y_1[t]\}_{t=1}^T\Big)  +T\epsilon_T \nonumber \\
&=  I\Big(m_2;\{h_{12}[t] X_2[t] \oplus Z_1[t]\}_{t=1}^T \Big| m_1 \Big) \nonumber \\
& \qquad+ I\Big(m_1;\{Y_1[t]\}_{t=1}^T\Big)  +T\epsilon_T\\
&=  I\Big(m_2;\{h_{11}[t]X_1[t] \oplus h_{12}[t] X_2[t] \oplus Z_1[t]\}_{t=1}^T \Big| m_1 \Big) \nonumber \\
& \qquad+ I\Big(m_1;\{Y_1[t]\}_{t=1}^T\Big)  +T\epsilon_T\\
&=  I\Big(m_2;\{Y_1[t]\}_{t=1}^T \Big| m_1 \Big) + I\Big(m_1;\{Y_1[t]\}_{t=1}^T\Big)  +T\epsilon_T \nonumber \\
&=  I\Big(m_1,m_2;\{Y_1[t]\}_{t=1}^T  \Big)   +T\epsilon_T\\
&\leq T\big(\log{q} - H(Z)\big)  +T\epsilon_T
\end{align} As the probability of error $p_{\text{error}}$ tends to zero, $\epsilon_T \rightarrow 0$ which yields $R_1 + R_2 \leq \log{q} - H(Z)$. Similar outer bounds hold for all receiver pairs $\ell$ and $k$. Comparing these to the achievable region in Theorem \ref{t:fieldalignment} yields the capacity region.
\end{proof}

\section{Conclusions}

In this paper, we proposed a new scheme, ergodic interference alignment, for time-varying interference channels.  
Overall, this scheme shows how much can be gained by coding over parallel interference channels. While in the Gaussian case, we can simply add up two well-matched channel outputs, in general, we can think about this alignment scheme as organizing the computations naturally provided by the channel. 

An interesting subject for future study is the inclusion of ergodic interference alignment into classical power allocation and Han-Kobayashi message-splitting strategies. That is, the optimal scheme will most likely have each transmitter split its message into several parts. Channel realizations will then have to be grouped according to which messages should be treated as noise, decoded, or aligned by each receiver.

\appendices

\section{Outer Bound} \label{s:outer}

We now develop an upper bound that is applicable when the receivers want to decode one or more messages over a time-varying Gaussian interference channel (the setting of Sections \ref{s:ergodic} and \ref{s:multicast}). The proof closely follows the multiple-access outer bound used in \cite{cj08}.

Assume, without loss of generality, that receiver $k$ wants to recover $m_1,\ldots, m_M$ and that receiver $n$ wants to recover (at least) $m_{M+1}$. Now, give receivers $k$ and $n$ the messages $m_{M+2},\ldots,m_L$ as genie-aided side information, which can only increase the rates $R_1,\ldots, R_{M+1}$. Both receivers can now completely remove the effects of $X_{M+2}[t], \ldots,X_{L}[t]$ from their observations. We also assume that $h_{k\ell}[t] \neq 0$ for $t =1,\ldots, T$. This occurs with probability $1$ for many fading distributions of interest. Finally, note that scaling the channel output at a receiver cannot change the capacity. Overall, we can assume that receivers $k$ and $n$ have access to the channel observations 
\begin{align}
\tilde{Y}_k[t] &= \sum_{\ell = 1}^{M+1} h_{k\ell}[t] X_\ell[t] + Z_k[t] \\
\tilde{Y}_n[t] &= \frac{h_{k,M+1}[t]}{h_{n,M+1}[t]}\bigg(\sum_{\ell = 1}^{M+1} h_{n\ell}[t] X_\ell[t] \bigg) + \tilde{Z}_n[t] 
\end{align} where
\begin{align}
\tilde{Z}_n[t] = \frac{h_{k,M+1}[t]}{h_{n,M+1}[t]} Z_n[t] \ . 
\end{align}

Since the receivers cannot cooperate, the capacity only depends on the noise marginals. It is useful to assume that the noise terms $Z_k[t]$ and $\tilde{Z}_n[t]$ are generated in a correlated fashion at each time step. Define
\begin{align}
\alpha[t] \triangleq \min\bigg(1,~ \frac{|h_{k,M+1}[t]|^2}{|h_{n,M+1}[t]|^2}\bigg)
\end{align} as well as the following independent noise processes
\begin{align}
\bar{Z}[t] &\sim \mathcal{CN}(0,\alpha[t])\\
\bar{Z}_k[t] &\sim \mathcal{CN}(0,1-\alpha[t])\\
\bar{Z}_n[t] &\sim \mathcal{CN}\bigg(0,\frac{| h_{k,M+1}[t] |^2}{| h_{n,M+1}[t] |^2}-\alpha[t]\bigg) \ .
\end{align} that are each i.i.d. across time. We combine these to create the correlated noise terms at the receivers
\begin{align}
Z_k[t] &= \bar{Z}[t] + \bar{Z}_k[t] \\
\tilde{Z}_n[t] &= \bar{Z}[t] + \bar{Z}_n[t]
\end{align}

We will also give $m_1,\ldots, m_M$ to receiver $n$ as genie-aided side-information. Define 
\begin{align}
\epsilon_T \triangleq 1 +p_{\text{error}} \sum_{k=1}^{M+1}R_k 
\end{align} where $p_{\text{error}}$ is the average probability of error. Via Fano's inequality, it follows that
\begin{align}
&T \sum_{k=1}^{M+1} R_k \nonumber \\
&\leq I\Big(m_{M+1};\big\{\tilde{Y}_n[t]\big\}_{t=1}^T, m_1,\ldots, m_M\Big) \nonumber \\
&\qquad +  I\Big(m_1,\ldots,m_M; \big\{\tilde{Y}_k[t]\big\}_{t=1}^T\Big) + T \epsilon_T \\ 
&=  I\Big(m_{M+1};\big\{\tilde{Y}_n[t]\big\}_{t=1}^T \big| m_1,\ldots, m_M\Big)  \nonumber\\
&\qquad + I\Big(m_1,\ldots,m_M; \big\{\tilde{Y}_k[t]\big\}_{t=1}^T\Big) + T \epsilon_T \\
&= I\Big(m_{M+1};\big\{h_{k,M+1}[t] X_{M+1}[t]  + \tilde{Z}_n[t]\big\}_{t=1}^T \big| m_1,\ldots, m_M\Big) \nonumber  \\
&\qquad + I\Big(m_1,\ldots,m_M; \big\{\tilde{Y}_k[t]\big\}_{t=1}^T\Big) + T \epsilon_T \\
&= I\Bigg(m_{M+1};\bigg\{ \sum_{\ell=1}^{M+1}h_{k\ell}[t] X_{\ell}[t]  + \tilde{Z}_n[t]\bigg\}_{t=1}^T \bigg| m_1,\ldots, m_M\Bigg) \nonumber  \\
&\qquad + I\Big(m_1,\ldots,m_M; \big\{\tilde{Y}_k[t]\big\}_{t=1}^T\Big) + T \epsilon_T
\end{align}

Now, we weaken the noise by giving receivers $\bar{Z}_k[t]$ and $\bar{Z}_n[t]$ as side information. Let 
\begin{align}
\bar{Y}_k[t] = \sum_{\ell=1}^{M+1}h_{k\ell}[t] X_{\ell}[t]  + \bar{Z}[t] \ .
\end{align} It follows that

\begin{align}
T \sum_{k=1}^{M+1} R_k &\leq I\Big(m_{M+1}; \big\{\bar{Y}_k[t]\big\}_{t=1}^T \big| m_1,\ldots,m_M \Big) \nonumber \\
&\qquad + I\Big(m_1,\ldots,m_M; \big\{\bar{Y}_k[t]\big\}_{t=1}^T\Big) + T \epsilon_T \\
&=  I\Big(m_1,\ldots,m_{M+1}; \big\{\bar{Y}_k[t]\big\}_{t=1}^T\Big) + T \epsilon_T  \ .
\end{align} Now, applying the usual steps, we can show that the mutual information expression is maximized by independent Gaussian inputs.

Assume that all transmitters employ a uniform power allocation across time. Specializing the upper bound above to the $K$-user interference channel from Section \ref{s:ergodic} (and taking $T \rightarrow \infty$), we get that
\begin{align}
R_{\ell} + R_k \leq \E\left[ \log\left(1 + \frac{( | h_{k\ell}|^2 + |h_{kk}|^2 )P}{\min\big(1, \frac{|h_{k\ell}|^2}{|h_{\ell \ell}|^2}\big)} \right) \right] \label{e:icupper}
\end{align} for all $k = 1,2,\ldots, K$ and $\ell \neq k$.

Specializing to the case in Section \ref{s:multicast}, where receiver $k$ wants the messages $m_{\ell}$ for $\ell \in \mathcal{S}_k$, we get that 
\begin{align}
&R_{i} + \sum_{\ell \in \mathcal{S}_k}R_\ell \label{e:multicastupper} \nonumber\\ 
&\leq \E\left[ \log\left(1 + \frac{\Big( | h_{ki}|^2 + \sum_{\ell\in \mathcal{S}_k}|h_{k\ell}|^2 \Big)P}{\min\big(1, \frac{|h_{ki}|^2}{|h_{ni}|^2}\big)} \right) \right]
\end{align} for all $i$ such that $i \in \mathcal{S}_n$ and $i \notin \mathcal{S}_k$.

\section{Computation Coding} \label{s:compute} 

We now review the computation coding scheme from \cite{ng07IT} for finite field channels. Assume that there are $\ell$ transmitters, each with a message $\mathbf{w}_{\ell} \in \mathbb{F}_q^{\kappa}$. Each transmitter maps its message into a length $\tau$ codeword $\mathbf{x}_\ell \in \mathbb{F}_q^\tau$.

Receiver $k$ observes a noisy linear combination of the codewords
\begin{align}
\mathbf{y}_k = \bigoplus_{\ell=1}^K h_{k\ell} \mathbf{x}_\ell \oplus \mathbf{z}_k
\end{align} where $\mathbf{z}_k$ is a noise vector whose elements are i.i.d. according to a distribution with entropy $H(Z)$. Each receiver would like to make an estimate $\mathbf{\hat{u}}_k$ of a linear equation of the messages
\begin{align}
\mathbf{u}_k = \bigoplus_{\ell=1}^K h_{k\ell} \mathbf{w}_\ell \ . 
\end{align} 
The following lemma states an achievable \textit{computation rate} for this setting.
\begin{lemma} \label{l:compute}
For any $\epsilon > 0$ and $\tau$ large enough, there exists a set of encoders and decoders such that all receivers can make estimates $\mathbf{\hat{u}}_k$ of the linear equations $\mathbf{u}_k$ with total probability of error
\begin{align}
\P\bigg(\bigcup_{k=1}^K \{ \mathbf{\hat{u}}_k \neq \mathbf{u}_k \} \bigg) < \epsilon 
\end{align} so long as the rate $\kappa/\tau$ satisfies
\begin{align}
\frac{\kappa}{\tau} < \frac{\log q - H(Z)}{\log q} \ . 
\end{align}
\end{lemma}
\begin{IEEEproof}
First, we find a linear code with generator matrix $\mathbf{G} \in \mathbb{F}_q^{\tau \times k}$ with rate $\kappa/\tau < (\log q - H(Z))/\log q$ and probability of error at most $\epsilon/K$ over the channel 
\begin{align}
\mathbf{y} = \mathbf{x} \oplus \mathbf{z}
\end{align} where $\mathbf{z}$ has the same distribution as $\mathbf{z}_k$ and $\mathbf{x} = \mathbf{G} \mathbf{w}$. Each encoder employs $\mathbf{G}$ to get $\mathbf{x}_\ell = \mathbf{G} \mathbf{w}_\ell$. As a result, each receiver sees
\begin{align}
\mathbf{y}_k & =  \bigoplus_{\ell=1}^K h_{k\ell} \mathbf{G} \mathbf{w}_\ell \oplus \mathbf{z}_k \\
&=  \mathbf{G}\bigg( \bigoplus_{\ell=1}^K h_{k\ell} \mathbf{w}_\ell\bigg) \oplus \mathbf{z}_k \\
&= \mathbf{G} \mathbf{u}_k \oplus \mathbf{z}_k 
\end{align} from which it can decode $\mathbf{u}_k$ with probability of error at most $\epsilon/K$. By the union bound, the total probability of error is at most $\epsilon$. 
\end{IEEEproof} Via standard cut-set arguments, it can also be shown that this is the computation capacity.

\section*{Acknowledgment}

The authors would like to thank the anonymous reviewers whose suggestions improved the presentation of this work.

\bibliographystyle{ieeetr}

\end{document}